\documentclass[11pt]{article}
\usepackage{fullpage}
\usepackage{url}
\usepackage{xspace}
\usepackage{color}
\usepackage{graphics}
\usepackage[dvips]{epsfig}

\usepackage{amsmath}
\usepackage{amssymb}
\usepackage{amsfonts}
\usepackage{graphicx}
\usepackage{algorithm2e}

\usepackage[small,compact]{titlesec}
\usepackage{times}

\newtheorem{theorem}{Theorem}[section]
\newtheorem{lemma}{Lemma}[section]

\newtheorem{proposition}{Proposition}[section]

\newtheorem{remark}{Remark}[section]

\newcommand{\qed}{\hfill $\Box$ \bigbreak}
\newenvironment{proof}{\noindent {\bf Proof.}}{\qed}

\newcommand{\remove}[1]{}

\begin{document}

\baselineskip  0.2in 
\parskip     0.1in 
\parindent   0.0in 

\title{{\bf Byzantine Gathering in Networks}\footnote{Partially supported by the European Regional Development Fund (ERDF) and the Picardy region under Project TOREDY.}}
\date{}
\newcommand{\inst}[1]{$^{#1}$}

\author{
S\'{e}bastien Bouchard\inst{1},
Yoann Dieudonn\'e\inst{1},
Bertrand Ducourthial\inst{2}\\
\inst{1} Laboratoire MIS \& Universit\'{e} de Picardie Jules Verne, Amiens, France.\\
\inst{2} Heudiasyc, CNRS \& Universit\'{e} de Technologie de Compi\`egne, Compi\`egne, France.\\
}

\maketitle
\vspace{-1cm}
\begin{abstract}

This paper investigates an open problem introduced in \cite{DieudonnePP14}. Two or more mobile agents start from different nodes of a network and have to accomplish the task of gathering which consists in getting all together at the same node at the same time. An adversary chooses the initial nodes of the agents and assigns a different positive integer (called label) to each of them. Initially, each agent knows its label but does not know the labels of the other agents or their positions relative to its own. Agents move in synchronous rounds and can communicate with each other only when located at the same node. Up to $f$ of the agents are Byzantine. A Byzantine agent can choose an arbitrary port when it moves, can convey arbitrary information to other agents and can change its label in every round, in particular by forging the label of another agent or by creating a completely new one. 

{\it What is the minimum number $\mathcal{M}$ of good agents that guarantees deterministic gathering of all of them, with termination?}


We provide exact answers to this open problem by considering the case when the agents initially know the size of the network and the case when they do not. In the former case, we prove $\mathcal{M}=f+1$ while in the latter, we prove $\mathcal{M}=f+2$. More precisely, for networks of known size, we design a deterministic algorithm gathering all good agents in any network provided that the number of good agents is at least $f+1$.  For networks of unknown size, we also design a deterministic algorithm ensuring the gathering of all good agents in any network but provided that the number of good agents is at least $f+2$. Both of our algorithms are optimal in terms of required number of good agents, as each of them perfectly matches the respective lower bound on $\mathcal{M}$ shown in \cite{DieudonnePP14}, which is of $f+1$ when the size of the network is known and of $f+2$ when it is unknown.


Perhaps surprisingly, our results highlight an interesting feature when put in perspective with known results concerning a relaxed variant of this problem in which the Byzantine agents cannot change their initial labels. Indeed under this variant $\mathcal{M}=1$ for networks of known size and $\mathcal{M}=f+2$ for networks of unknown size. Following this perspective, it turns out that when the size of the network is known, the ability for the Byzantine agents to change their labels significantly impacts the value of $\mathcal{M}$. However, the relevance for $\mathcal{M}$ of such an ability completely disappears in the most general case where the size of the network is unknown, as $\mathcal{M}=f+2$ regardless of whether Byzantine agents can change their labels or not.



\vspace{2ex}

\end{abstract}

\vfill

\vfill

\thispagestyle{empty}
\setcounter{page}{0}
\pagebreak

\section{Introduction}

\subsection{Context}Gathering is one of the most fundamental tasks in the field of distributed and mobile systems in the sense that, the ability to gather is in fact a building block to achieve more complex cooperative works. Loosely speaking, the task of gathering consists in ensuring that a group of mobile entities, initially located in different places, ends up meeting at the same place at the same time. These mobile entities, hereinafter called agents, can vary considerably in nature ranging from human beings and robots to animals and software agents. The environment in which the agents are supposed to evolve can vary considerably as well: it may be a terrain, a network modeled as a graph, a three-dimensional space, etc. We can also consider that the sequences of instructions followed by the agents in order to ensure their gathering are either deterministic or randomized. 

In this paper, we consider the problem of gathering in a deterministic way in a network modeled as a graph. Thus, the agents initially start from different nodes of the graph and have to meet at the same node by applying deterministic rules. We assume that among the agents, some are Byzantine. A Byzantine agent is an agent subject to unpredictable and arbitrary faults. For instance such an agent may choose to never stop or to never move. It may also convey arbitrary information to the other agents, etc. The case of Byzantine fault is very interesting because it is the worst fault that can occur to agents. As a consequence, gathering in such a context is challenging.

\subsection{Model and problem}
\label{sub:subm}

The distributed system considered in this paper consists of a group of mobile agents that are initially placed by an adversary at arbitrary but distinct nodes of a network modeled as a finite, connected, undirected graph $G=(V,E)$. We assume that $|V|=n$. In the sequel $n$ is also called the size of the network. Two assumptions are made about the labelling of the two main components of the graph that are nodes and edges. The first assumption is that nodes are anonymous i.e., they do not have any kind of labels or identifiers allowing them to be distinguished from one another. The second assumption is that edges incident to a node $v$ are locally ordered with a fixed port numbering ranging from $0$ to $deg(v)-1$ where $deg(v)$ is the degree of $v$. Therefore, each edge has exactly two port numbers, one for each of both nodes it links. The port numbering is not supposed to be consistent: a given edge $(u,v)\in E$ may be the $i$-th edge of $u$ but the $j$-th edge of $v$, where $i\ne j$. These two assumptions are not fortuitous. The primary motivation of the first one is that if each node could be identified by a label, gathering would become quite easy to solve as it would be tantamount to explore the graph (via e.g. a breadth-first search) and then meet in the node having the smallest label. While the first assumption is made so as to avoid making the problem trivial, the second assumption is made in order to avoid making the problem impossible to solve. Indeed, in the absence of a way allowing an agent to distinguish locally the edges incident to a node, gathering could be proven as impossible to solve deterministically in view of the fact that some agents could be precluded from traversing some edges and visit some parts of the graph.

An adversary chooses the starting nodes of the agents. The starting nodes are chosen so that there are not two agents sharing initially the same node. At the beginning, an agent has a little knowledge about its surroundings: it does not know either the graph topology, or the number of other agents, or the positions of the others relative to its own. Still regarding agents' knowledge, we will study two scenarios: one in which the agents initially know the parameter $n$ and one in which the agents do not initially know this parameter or even any upper bound on it.

Time is discretized into an infinite sequence of rounds. In each round, every agent, which has been previously woken up (this notion is detailed in the next paragraph), is allowed to stay in place at its current node or to traverse an edge according to a deterministic algorithm. The algorithm is the same for all agents: only the input, whose nature is specified further in the subsection, varies among agents.

Before being woken up, an agent is said to be dormant. A dormant agent may be woken up only in two different ways: either by the adversary that wakes some of the agents at possibly different rounds, or as soon as another agent enters the starting node of the dormant agent. We assume that the adversary wakes up at least one agent. 

When an agent is woken up in a round $r$, it is told the degree of its starting node. As mentioned above, in each round $r'\geq r$, the executed algorithm can ask the agent to stay idle or to traverse an edge. In the latter case, this takes the following form: the algorithm ask the agent, located at node $u$, to traverse the edge having port number $i$, where $0 \leq i < deg(u)-1$. Let us denote by $(u,v)\in E$ this traversed edge. In round $r'+1$, the agents enters node $v$: it then learns the degree $deg(v)$ as well as the local port number $j$ of $(u,v)$ at node $v$ (recall that in general $i\ne j$). An agent cannot leave any kind of tokens or markers at the nodes it visits or the edges it traverses.

In the beginning, the adversary also assigns a different positive integer (called label) to each agent. Each agent knows its label but does not know the labels of the other agents. When several agents are at the same node in the same round, they see the labels of the other agents and can exchange all the information they currently have. This exchange is done in a ``shouting'' mode in one round: all the exchanged information becomes common knowledge for agents that are currently at the node. On the other hand when two agents are not at the same node in the same round they cannot see or talk to each other: in particular, two agents traversing simultaneously the same edge but in opposite directions, and thus crossing each other on the same edge, do not notice this fact. In every round, the input of the algorithm executed by an agent $a$ is made up of the label of agent $a$ and the up-to-date memory of what agent $a$ has seen and learnt since its waking up. Note that in the absence of a way of distinguishing the agents, the gathering problem would have no deterministic solution in some graphs. This is especially the case in a ring in which at each node the edge going clockwise has port number $0$ and the edge going anti-clockwise has port $1$: if all agents are woken up in the same round and start from different nodes, they will always have the same input and will always follow the same deterministic rules leading to a situation where the agents will always be at distinct nodes no matter what they do.
 
Within the team, it is assumed that up to $f$ of the agents are Byzantine. The parameter $f$ is known to all agents. A Byzantine agent has a high capacity of nuisance: it can choose an arbitrary port when it moves, can convey arbitrary information to other agents and can change its label in every round, in particular by forging the label of another agent or by creating a completely new one. All the agents that are not Byzantine are called good.  We consider the task of $f$-Byzantine gathering which is stated as follows. The adversary wakes up at least one good agent and all good agents must eventually be in the same node in the same round, simultaneously declare termination and stop, provided that there are at most $f$ Byzantine agents. Regarding this task, it is worth mentioning that we cannot require the Byzantine agents to cooperate as they may always refuse to be with some agents. Thus, gathering all good agents with termination is the strongest requirement we can make in such a context.

What is the minimum number $\mathcal{M}$ of good agents that guarantees $f$-Byzantine gathering? 

At first glance, the question might appear as not being really interesting since, after all, the good agents might always be able to gather in some node, regardless of the number of Byzantine agents evolving in the graph. However, this is not the case as pointed out by the study that introduced this question in \cite{DieudonnePP14}. More specifically, when this size is initially known to the agents, the authors of this study described a deterministic algorithm gathering all good agents in any network provided that there are at $2f+1$ of them, and gave a lower bound of $f+1$ on $\mathcal{M}$ by showing that if the number of good agents is not larger than $f$, then there are some graphs in which the good agents are not able to gather deterministically with termination. When the size of the network is unknown, they did a similar thing but with different bounds: they gave an algorithm working for a team including at least $4f+2$ good agents, and showed a lower bound of $f+2$ on $\mathcal{M}$. However, the question of what the tight bounds are was left as an open problem.

\subsection{Our results}

In this paper, we solve this open problem by proving that the lower bounds of $f+1$ and $f+2$ on $\mathcal{M}$, shown in \cite{DieudonnePP14}, are actually also upper bounds respectively when the size of the network is known and when it is unknown. More precisely, we design deterministic algorithms allowing to gather all good agents provided that the number of good agents is at least $f+1$ when the size of the network is initially known to agents, and at least $f+2$ when this size is initially unknown.

Perhaps surprisingly, our results highlight an interesting feature when put in perspective with results concerning a relaxed variant of this problem (also introduced in~\cite{DieudonnePP14}) in which the Byzantine agents cannot change their initial labels. Indeed under this variant $\mathcal{M}=1$ for networks of known size and $\mathcal{M}=f+2$ for networks of unknown size\footnote{The proof that both of these values are enough, under their respective assumptions regarding the knowledge of the network size, relies on algorithms using a mechanism of blacklists that are, informally speaking, lists of labels corresponding to agents having exhibited an ``inconsistent" behavior. Of course, in the context of our paper, we cannot use such blacklists as the Byzantine agents can change their labels and in particular steal the identities of good agents.}. Following this perspective, it turns out that when the size of the network is known, the ability for the Byzantine agents to change their labels significantly impacts the value of $\mathcal{M}$. However, the relevance for $\mathcal{M}$ of such an ability completely disappears in the most general case where the size of the network is unknown, as $\mathcal{M}=f+2$ regardless of whether Byzantine agents can change their labels or not.


\subsection{Related works}

Historically, the first mention of the gathering problem appeared in \cite{Schelling} under the appellation of rendezvous problem. Rendezvous is the term which is usually used when the studied task of gathering is restricted to a team of exactly two agents. From this publication until now, the problem has been extensively studied so that there is henceforth a huge literature about this subject. This is mainly due to the fact that there is a lot of alternatives for the combinations we can make when approaching the problem, e.g., by playing on the environment in which the agents are supposed to evolve, the way of applying the sequences of instructions (i.e., deterministic or randomized) or the ability to leave some traces in the visited locations, etc. Naturally, in this paper we are more interested in the research works that are related to deterministic gathering in networks modeled as graphs. This is why we will mostly dwell on this scenario in the rest of this subsection. However, for the curious reader wishing to consider the matter in greater depth, we invite him to consult \cite{CieliebakFPS12,AgmonP06,IzumiSKIDWY12} that address the problem in the plane via various scenarios, especially in a system affected by the occurrence of faults or inaccuracies for the last two references. Regarding randomized rendezvous, a good starting point is to go through \cite{Alpern02,Alpern03,KranakisKR06}.

Concerning the context of this paper, the closest work to ours is obviously \cite{DieudonnePP14}. Nonetheless, in similar settings but without Byzantine agents, there are some papers that should be cited here. This is in particular the case of \cite{DessmarkFKP06} in which the author presented a deterministic protocol for solving the rendezvous problem, which guarantees a meeting of the two involved agents after a number of rounds that is polynomial in the size $n$ of the graph, the length $l$ of the shorter of the two labels and the time interval $\tau$ between their wake-up times. As an open problem, the authors ask whether it is possible to obtain a polynomial solution to this problem which would be independent of $\tau$. A positive answer to this question was given, independently of each other, in \cite{KowalskiM08} and \cite{Ta-ShmaZ14}. While these algorithms ensure rendezvous in polynomial time (i.e., a polynomial number of rounds), they also ensure it at polynomial cost since the cost of a rendezvous protocol is the number of edge traversals that are made by the agents until meeting and since each agent can make at most one edge traversal per round. However, it should be noted that despite the fact a polynomial time implies a polynomial cost, the reciprocal is not always true as the agents can have very long waiting periods sometimes interrupted by a movement. Thus these parameters of cost and time are not always linked to each other. This was highlighted in \cite{MillerP14a} where the authors studied the tradeoffs between cost and time for the deterministic rendezvous problem. More recently, some efforts have been dedicated to analyse the impact on time complexity of rendezvous when in every round the agents are brought with some pieces of information by making a query to some device or some oracle, see, e.g., \cite{DKU14,MillerP14b}. Along with the works aiming at optimizing the parameters of time and/or cost of rendezvous, some other works have examined the amount of memory that is required to achieve deterministic rendezvous e.g., in \cite{FraigniaudP08,FraigniaudP13} for tree networks and in \cite{CzyzowiczKP12} for general networks.

All the aforementioned studies that are related to gathering in graphs take place in a synchronous scenario i.e., a scenario in which the agents traverse the edges in synchronous rounds. Some efforts have been also dedicated to the scenario in which the agents move asynchronously: the speed of agents may then vary and is controlled by the adversary. For more details about rendezvous under such a context, the reader is referred to \cite{MarcoGKKPV06,CzyzowiczPL12,DieudonnePV13,GuilbaultP13} for rendezvous in finite graphs and \cite{BampasCGIL10,CollinsCGL10} for rendezvous in infinite grids.

Aside from the gathering problem, our work is also in conjunction with the field of fault tolerance via the assumption of Byzantine faults to which some agents are subjected. First introduced in \cite{PeaseSL80}, a Byzantine fault is an arbitrary fault occurring in an unpredictable way during the execution of a protocol. Due to its arbitrary nature, such a fault is considered as the worst fault that can occur. Byzantine faults have been extensively studied for ``classical" networks i.e., in which the entities are fixed nodes of the graph (cf., e.g., the book \cite{Lynch96} or the survey \cite{BarborakM93}). To a lesser extend, the occurrence of Byzantine faults has  been also studied in the context of mobile entities evolving in the plane, cf. \cite{AgmonP06,DefagoGMP06}. Prior to our work, gathering in arbitrary graphs in presence of Byzantine agents was considered only in \cite{DieudonnePP14}. As mentioned in the previous section, it is proven in \cite{DieudonnePP14} that the minimum number $\mathcal{M}$ of good agents that guarantees $f$-Byzantine gathering is precisely $1$ for networks of known size and $f+2$ for networks of unknown size, provided that the Byzantine agents cannot lie about their labels. The proof that both of these values are enough, under their respective assumptions regarding the knowledge of the network size, relies on algorithms using a mechanism of blacklists that are, informally speaking, lists of labels corresponding to agents having exhibited an ``inconsistent" behavior. Of course, in the context of our paper, we cannot use such blacklists as the Byzantine agents can change their labels and in particular steal the identities of good agents. 

\section{Preliminaries}
\label{sec:pre}
Throughout the paper, 
the number of nodes of a graph is called its size.
In this section we {present two procedures}, that will be used as building blocks in our algorithms. 
The aim of both of them is graph exploration, i.e., visiting all nodes of the graph by a single agent. 
The first procedure, based on universal exploration sequences (UXS), is a corollary of the result of Reingold \cite{Reingold08}. Given any positive integer $N$, this procedure allows the agent to traverse all nodes of any graph of size at most $N$,
starting from any node of this graph, using $P(N)$ edge traversals, where $P$ is some polynomial. After entering a node of degree $d$ by some port $p$,
the agent can compute the port $q$ by which it has to exit; more precisely $q=(p+x_i)\mod d$, where $x_i$ is the corresponding term of the UXS of length $P(N)$.

The second procedure \cite{ChalopinDK10} needs no assumption on the size of the network but 
it is performed by an agent using a fixed token placed at a node of the graph. It works in time polynomial in the size of the graph.
(It is well known that a terminating exploration even of all anonymous rings of unknown size by a single agent without a token is impossible.)
In our applications the roles of the token and of the exploring agent will be played by agents or by groups of agents.
At the end of this second procedure, the agent has visited all nodes and determined a BFS tree of the underlying graph.

We call the first procedure $EXPLO(N)$ and the second procedure $EST$, for {\em exploration with a stationary token}. We denote by $T(EXPLO(n))$ the execution time of procedure $EXPLO$ with parameter $n$ (note that $T(EXPLO(n))=P(n)+1$).
We denote by $T(EST(N))$ the maximum time of execution of the procedure $EST$
in a graph of size at most $N$.

\section{Known graph size}
\label{sec:sec1}

This section aims at proving the following theorem
\vspace{-0.3cm}
\begin{theorem}
\label{theo:theo1}
Deterministic $f$-Byzantine gathering of $k$ good agents is possible in any graph of known size if, and only if $k\geq f+1$.
\end{theorem}
\vspace{-0.3cm}

As mentioned in Subsection~\ref{sub:subm}, we know from \cite{DieudonnePP14} that:
\vspace{-0.3cm}
\begin{theorem}[\cite{DieudonnePP14}]
\label{theo:theo2}
Deterministic $f$-Byzantine gathering of $k$ good agents is not possible in some graph of known size if $k\leq f$.
\end{theorem}
\vspace{-0.3cm}
Thus, to prove Theorem~\ref{theo:theo1}, it is enough to show the following theorem.
\vspace{-0.3cm}
\begin{theorem}
\label{theo:theo3}
Deterministic $f$-Byzantine gathering of $k$ good agents is possible in any graph of known size if $k\geq f+1$.
\end{theorem}
\vspace{-0.3cm}
Hence, the rest of this section is devoted to proving Theorem~\ref{theo:theo3}. To do so, we show a deterministic algorithm that gathers all good agents in an arbitrary network of known size, provided there are at least $f+1$ of them.

Before presenting the algorithm, we first give the high level idea which is behind it. Let us assume an ideal situation in which each agent would have as input, besides its label and the network size $n$, a parameter $\rho=(G^*,L^*)$ corresponding to the initial configuration of the agents in the graph such that: 
\begin{itemize}
\item $G^*$ represents the $n$-node graph with all port numbers, in which each node are assigned an identifier belonging to $\{1,\cdots,n\}$. The node identifiers are pairwise distinct. Note that the representation $G^*$ contains more information than there is in the actual graph $G$ as it also includes node identifiers which do not exist in $G$.

\item $L^*=\{(v_1,l_1),(v_2,l_2),\cdots,(v_k,l_k)\}$ where $(v_i,l_i)\in L^*$ iff there is a good agent having label $l_i$ which is initially placed in $G$ at the node corresponding to $v_i$ in $G^*$. Remark that $k\geq f+1$.
\end{itemize}

Let us also assume that all the agents in the graph are woken up at the same time by the adversary. In such ideal situation, gathering all good agents can be easily achieved by ensuring that each agent moves towards the node $v$ where the agent having the smallest label is located. Each agent can indeed do that by using the knowledge of $\rho=(G^*,L^*)$ and its own label. Of course, all the good agents do not necessarily reach node $v$ at the same time. However, once at node $v$, each agent can estimate the remaining time which is required to wait in order to be sure that all good agents are at node $v$: again this estimation can be computed using $\rho=(G^*,L^*)$ and the fact that all agents are woken up in the same round.
Unfortunately, the agents are not in such ideal situation. First, every agent is not necessarily woken up by the adversary, and for those that are woken by the adversary, this is not necessarily in the same round. Second, the agents do not have configuration $\rho$ as input of the algorithm. In our algorithm we cope with the first constraint by requiring the first action to be a traversal of the entire graph (using procedure $EXPLO(n)$) which allows to wake up all encountered agents that are still dormant. In this way, the agents are ``almost synchronized'' as the delay between the starting times of any two agents is at most $T(EXPLO(n))$: the waiting time periods can be adjusted regarding this maximum delay. The second constraint i.e., the non-knowledge of $\rho$, is more complicated to deal with. To handle the lack of information about $\rho$, agents make successive assumptions about it that are ``tested'' one by one. More precisely, let $\mathcal{P}$ be the recursively enumerable set of all the configurations $\rho_i=(G_i^*,L_i^*)$ such that $G_i^*$ is a connected $n$-node graph and $|L_i^*|\geq f+1$. Let $\Theta=(\rho_1,\rho_2,\rho_3,\cdots)$ be a fixed enumeration of $\mathcal{P}$ (all good agents agree on this enumeration). Each agent proceeds in phases numbered $1,2,3,\cdots$. In each phase $i$, an agent supposes that $\rho=\rho_i$ and, similarly as in the ideal situation, tries to go to the node which is supposed to correspond to node $v$, where $v$ is the node where the agent having the smallest label is initially located (according to  $\rho_i$). For some reasons detailed in the algorithm (refer to the description of state {\tt setup}), when $\rho_i\ne\rho$ some agents may be unable to make such a motion. As a consequence, these agents will consider that, rightly, $\rho_i\ne\rho$. On the other hand, whether $\rho_i\ne\rho$ or not, some other good agents may reach a node for which they had no reason to think it is not $v$ (and thus $\rho_i\ne\rho$). The danger here is that when reaching the supposed node $v$ these successful agents could see all the $|L_i^*|$ labels of $\rho_i$ (with the possible ``help'' of some Byzantine agents). At this point, it may be tempting to consider that gathering is over but this could be wrong especially in the case where $\rho_i\ne\rho$ and some good agents did not reach a supposed node $v$ in phase $i$. To circumvent this problem, the idea is to get the good agents thinking that $\rho_i=\rho$ to fetch the (possible) others for which $\rho_i\ne\rho$ via a traversal of the entire graph using procedure $EXPLO(n)$ (refer to the description of state {\tt tower}). To allow this, an agent for which $\rho_i\ne\rho$ will wait a prescribed amount of rounds in order to leave enough time for possible good agents to fetch it (refer to the description of state {\tt wait-for-a-tower}). For our purposes, it is important to prevent the agents from being fetched any old how by any group, especially those containing only Byzantine agents. Hence our algorithm is designed in such a way that within each phase at most one group, called a \emph{tower} and made up of at least $f+1$ agents, will be unambiguously recognized as such and be allowed to fetch the other agents via an entire traversal of the graph (this guarantee principally results from the rules that are prescribed in the description of state {\tt tower builder}). When a tower has finished the execution of procedure $EXPLO(n)$ in some phase $i$, our algorithm guarantees that all good agents are together and declare gathering is over at the same time (whether the assumed configuration $\rho_i$ corresponds to the real initial configuration or not). On the other hand, in every phase $i$, if a tower is not created or ``vanishes'' (because there at not at least $f+1$ agents inside of it anymore) before the completion of its traversal, no good agent will declare that gathering is over in phase $i$.
In the worst case, the good agents will have to wait until assuming a good hypothesis about the real initial configuration, in order to witness the creation of a tower which will proceed to an entire traversal of the network (and thus declare gathering is over).

We now give a detailed description of the algorithm.

{\bf Algorithm Byz-Known-Size} with parameter $n$ (know size of the graph)

The algorithm is made up of two parts. The first part aims at ensuring that all agents are woken up before proceeding to the second part which is actually the heart of the algorithm.

\paragraph{Part~1.} As soon as an agent is woken up by the adversary or another agent, it starts proceeding to a traversal of the entire graph and wakes up all encountered agents that are still dormant. This is done using procedure $EXPLO(n)$ where $n$ is the size of the network which is initially known to all agents. Once the execution of $EXPLO(n)$ is accomplished, the agent backtracks to its starting node by traversing all edges traversed in $EXPLO(n)$ in the reverse order and the reverse direction.

\paragraph{Part~2.} In this part, the agent works in phases numbered $1,2,3,\cdots$. During the execution of each phase, the agent can be in one of the following five states: {\tt setup}, {\tt tower builder}, {\tt tower}, {\tt wait-for-a-tower}, {\tt failure}. Below we describe the actions of an agent $A$ in each of the states as well as the transitions between these states within phase $i$. We assume that in every round agent $A$ tells the others (sharing the same node as agent $A$) in which state it is. In some states, the agent will be required to tell more than just its current state: we will mention it in the description of these states. Moreover, in the description of every state X, when we say ``agent $A$ transits to state Y", we exactly mean agent $A$ remains in state X until the end of the current round and is in state Y in the following round. Thus, in each round of this part, agent $A$ is always exactly in one state.

At the beginning of phase $i$, agent $A$ enters state {\tt setup}.

  \noindent
 {\bf State} {\tt setup}.
 
Let $\rho_i$ be the $i$-th configuration of enumeration $\Theta$ (refer to above). If the label $l$ of agent $A$ is not in $\rho_i$, then it transits to state {\tt wait-for-a-tower}. Otherwise, let $X$ be the set of the shortest paths in $\rho_i$ leading from the node containing the agent having label $l$, to the node containing the smallest label of the supposed configuration. Each path belonging to $X$ is represented as the corresponding sequence of port numbers. Let $\pi$ be the lexicographically smallest path in $X$ (the lexicographic order can be defined using the total order on the port numbers). Agent $A$ follows path $\pi$ in the real network. If 
, following path $\pi$, agent $A$ has to leave by a port number that does not exist in the node where it currently resides, then it transits to state {\tt wait-for-a-tower}. In the same way, it also transits to state {\tt wait-for-a-tower} if, following path $\pi$, agent $A$ enters at some point a node by a port number which is not the same as that of path $\pi$. Once path $\pi$ is entirely followed by agent $A$, it transits to state {\tt tower builder}.
  \vspace*{0.2cm}
 
  \noindent
 {\bf State} {\tt tower builder}.

When in state {\tt tower builder}, agent $A$ can be in one of the following three substates: {\tt yellow}, {\tt orange}, {\tt red}. In all of these substates the agent does not make any move: it stays at the same node denoted by $v$. At the beginning, agent $A$ enters substate {\tt yellow}. By misuse of language, in the rest of this paper we will sometimes say that an agent ``is {\tt yellow}" instead of ``is in substate {\tt yellow}". We will also use the same kind of shortcut for the two other colors. In addition to its state, we also assume that in every round agent $A$ tells the others in which substate it is.

{\bf Substate} {\tt yellow}

Let $k$ be the number of labels in configuration $\rho_i$.
Agent $A$ waits $T(EXPLO(n))+n$ rounds. If during this waiting period, there are at some point at least $k$ {\tt orange} agents at node $v$ then agent $A$ transits to substate {\tt red}. Otherwise, if at the end of this waiting period there are not at least $k$ agents residing at node $v$ such that each of them is either {\tt yellow} or {\tt orange}, then agent $A$ transits to state {\tt wait-for-a-tower}, else it transits to substate {\tt orange}.
  
{\bf Substate} {\tt orange}

Agent $A$ waits at most $T(EXPLO(n))+n$ rounds to see the occurrence of one of the following two events. The first event is that there are not at least $k$ agents residing at node $v$ such that each of them is either {\tt yellow} or {\tt orange}. The second event is that there are at least $k$ {\tt orange} agents residing at node $v$. Note that the two events cannot occur in the same round. If during this waiting period, the first (resp. second) event occurs, then agent $A$ transits to state {\tt wait-for-a-tower} (resp. substate {\tt red}). If at the end of the waiting period, none of these events has occurred, then agent $A$ transits to substate {\tt wait-for-a-tower}.

{\bf Substate} {\tt red}

Agent $A$ waits $T(EXPLO(n))+n$ rounds. If at each round of this waiting period there are at least $k$ {\tt red} agents at node $v$, then at the end of the waiting period, agent $A$ transits to state {\tt tower}. Otherwise, there is a round during the waiting period in which there are not at least $k$ {\tt red} agents at node $v$: agent $A$ then transits to state {\tt wait-for-a-tower} as soon as it notices this fact.

  \noindent
 {\bf State} {\tt tower}.

Agent $A$ can enter state {\tt tower} either from state {\tt tower builder} or state {\tt wait-for-a-tower}. While in this state, agent $A$ will execute all or part of procedure $EXPLO(n)$. In both cases we assume that, in every round, agent $A$ tells the others the edge traversal number of $EXPLO(n)$ it has just made (in addition to its state). We call this number the index of the agent. Below, we distinguish and detail the two cases.

When agent $A$ enters state {\tt tower} from state {\tt tower builder}, it starts executing procedure $EXPLO(n)$. In the first round, its index is $0$. Just after making the $j$-th edge traversal of $EXPLO(n)$, its index is $j$.  
Agent $A$ carries out the execution of $EXPLO(n)$ until its term, except if at some round of the execution the following condition is not satisfied, in which case agent $A$ transits to state {\tt failure}. Here is the condition: the node where agent $A$ is currently located contains a group $\mathcal{S}$ of at least $f+1$ agents in state {\tt tower} having the same index as agent $A$. $\mathcal{S}$ includes agent $A$ but every agent that is in the same node as agent $A$ is not necessarily in $\mathcal{S}$. If at some point this condition is satisfied and the index of agent $A$ is equal to $P(n)$, which is the total number of edge traversals in $EXPLO(n)$ (refer to Section~\ref{sec:pre}), then agent $A$ declares that gathering is over. 

When agent $A$ enters state {\tt tower} from state {\tt wait-for-a-tower}, it has just made the $s$-th edge traversal of $EXPLO(n)$ for some $s$ (cf. state {\tt wait-for-a-tower}) and thus, its index is $s$. Agent $A$ executes the next edge traversals
 i.e., the $s+1$-th, $s+2$-th, $\cdots$, and then its index is successively $s+1$, $s+2$, etc. Agent $A$ carries out this execution until the end of procedure $EXPLO(n)$, except if the same condition as above is not fulfilled at some round of the execution of the procedure, in which case agent $A$ also transits to state {\tt failure}. As in the first case, if at some point the node where agent $A$ is currently located contains a group $\mathcal{S}$ of at least $f+1$ agents in state {\tt tower} having an index equal to $P(n)$, then agent $A$ declares that gathering is over.

\noindent
 {\bf State} {\tt wait-for-a-tower}.

Agent $A$ waits at most $5T(EXPLO(n))+4n$ rounds to see the occurrence of the following event: the node where it is currently located contains a group of at least $f+1$ agents in state {\tt tower} having the same index $t$. 
If during this waiting period, agent $A$ sees such an event, we distinguish two cases. If $t<P(n)$, then it makes the $t+1$-th edge traversal of procedure $EXPLO(n)$ and transits to state {\tt tower}. If $t=P(n)$, then it declares that gathering is over.

Otherwise, at the end of the waiting period, agent $A$ has not seen such an event, and thus it transits to state {\tt failure}.

\noindent
 {\bf State} {\tt failure}.
Agent $A$ backtracks to the node where it was located at the beginning of phase $i$. To do this, agent $A$ traverses in the reverse order and the reverse direction all edges it has traversed in phase $i$ before entering state {\tt failure}. Once at its starting node, agent $A$ waits $10T(EXPLO(n))+9n-p$ rounds where $p$ is the number of elapsed rounds between the beginning of phase $i$ and the end of the backtrack it has just made. At the end of the waiting period, 
phase $i$ is over. In the next round, agent $A$ will start phase $i+1$.

\subsection{Proof of correctness}

The proof of correctness is made up of seven lemmas and two propositions. The validity of Algorithm Byz-Known-Size (and by extension Theorem~\ref{theo:theo3}) follows from Lemmas~\ref{lem:6} and~\ref{lem:7}. However to prove both these lemmas, we first need to establish the following two propositions and to prove the following five lemmas.

Part~1 of Algorithm Byz-Known-Size consists in executing procedure $EXPLO(n)$ and then traversing all edges traversed in $EXPLO(n)$ in the reverse order and the reverse direction. In view of the fact that procedure $EXPLO(n)$ allows an agent to visit all nodes of the graph, we know that the delay between the starting rounds of any two good agents is at most $T(EXPLO(n))$. Hence we get the following proposition.

\begin{proposition}
\label{prop:prop1}
Let $A$ and $B$ be two good agents. The delay between the starting rounds of agents $A$ and $B$ is at most $T(EXPLO(n))$.
\end{proposition}

According to Part~1 of Algorithm Byz-Known-Size and the rules of state {\tt failure}, we have the following proposition.

\begin{proposition}
\label{prop:prop2}
At the beginning of every phase it executes, a good agent is at the node where it was woken up.
\end{proposition}

\begin{lemma}
\label{lem:lem1}
Let $Q(n)=10T(EXPLO(n))+9n$. For every good agent $A$ and every positive integer $i$, if at some point agent $A$ starts executing the $i$-th phase of Algorithm Byz-Known-Size, then either it will spend exactly $Q(n)$ rounds executing the $i$-th phase or it will declare that gathering is over after having spent at most $Q(n)$ rounds in the $i$-th phase.
\end{lemma}

\begin{proof}
Let $A$ be an agent that starts executing the $i$-th phase of Algorithm Byz-Known-Size in round $r$ and assume that agent $A$ does not declare that gathering is over by round $r+Q(n)-1$. To prove the lemma, it is then enough to prove that agent $A$ starts executing the $i+1$-th phase of Algorithm Byz-Known-Size in round $r+Q(n)$.

Let us first assume by contradiction that agent $A$ declares gathering is over when executing the $i$-th phase of Algorithm Byz-Known-Size. In view of the fact that agent $A$ does not declare that gathering is over by round $r+Q(n)-1$, we know that this declaration occurs in round $r+Q(n)$ at the earliest. However, according to Algorithm Byz-Known-Size, and especially the maximum duration of each state within any given phase (in particular a good agent cannot spend more than $n$ rounds in state {\tt setup} of any given phase in view of Proposition~\ref{prop:prop2}), agent $A$ cannot spend more than $9T(EXPLO(n))+8n$ rounds in phase $i$ without entering state {\tt failure}. Hence agent $A$ enters state {\tt failure} of phase $i$ before round $r+Q(n)$. But according to Algorithm Byz-Known-Size, once agent $A$ is in state {\tt failure} of phase $i$, it cannot reach the states {\tt tower} and {\tt wait-for-a-tower} of phase $i$, which are the only two states wherein an agent can declare that gathering is over. We then get a contradiction with the fact that agent $A$ declares gathering is over when executing the $i$-th phase of Algorithm Byz-Known-Size.

As a consequence, we know that agent $A$ does not declare gathering is over when executing the $i$-th phase of Algorithm Byz-Known-Size. According to the algorithm, it then ends up entering state {\tt failure} of phase $i$ after a finite number $x$ of rounds which is upper bounded by $9T(EXPLO(n))+8n$ as mentioned above. When entering state {\tt failure}, agent $A$ starts backtracking to the node $v$ where it was located at the beginning of phase $i$. Since before switching to state {\tt failure} of phase $i$, agent $A$ follows a path made up of at most $n-1$ edges and can proceed to at most one (possibly truncated) execution of procedure $EXPLO(n)$, backtracking to $v$, in state {\tt failure} of phase $i$, takes $y\leq T(EXPLO(n))+n$ rounds. When the backtrack is done, agent $A$ has spent $p=x+y\leq 10T(EXPLO(n))+9n$ rounds in phase $i$ and starts waiting $10T(EXPLO(n))+9n-p$ rounds. The end of the waiting period is reached in some round $r'$ when agent $A$ has spent exactly $10T(EXPLO(n))+9n$ rounds in phase $i$: according to Algorithm Byz-Known-Size, agent $A$ starts executing phase $i+1$ in round $r'+1$, which proves the lemma.
 
\end{proof}

Before continuing, we need to introduce some slight additional notions in order to facilitate the presentation of other lemmas. For any positive integer $i$, we say that a good agent $A$ {\emph tests} configuration $\rho_i$ when it executes the $i$-th phase of Algorithm Byz-Known-Size. We also say that there is a tower $\mathcal{T}_j$ at node $v$ in round $t$ if, and only if, there are at least $f+1$ agents, at node $v$ in round $t$, which are in state {\tt tower} and have index $j$. The members that make up $\mathcal{T}_j$ at node $v$ in round $t$ are then all the agents, which are at node $v$ in round $t$, being in state {\tt tower} and having index $j$. 

\begin{remark}
\label{rem:rem1}
Since there are at least $f+1$ agents in every tower, there is at least one good agent in every tower as there are at most $f$ Byzantine agents in the network.
\end{remark}

\begin{lemma}
\label{lem:lem2}
Let $T_j$ be a tower located at node $v$ in round $t$. There is at least one good agent in $T_j$ which has been in state {\tt tower} since round $t-j$.
\end{lemma}

\begin{proof}
Assume by contradiction that there is no good agent in tower $T_j$ which has been in state {\tt tower} since round $t-j$. By definition of a tower, every good agent, which is in tower $T_j$, is in state {\tt tower} and has index $j$ in round $t$. Let $A$ be a good agent in $T_j$. Note that agent $A$ exists in view of Remark~\ref{rem:rem1}. The last time agent $A$ decided to transit to state {\tt tower} before round $t$ was necessarily from state {\tt wait-for-a-tower}. Indeed, if that was not the case, that would imply that the last round $r<t$ when agent $A$ decided to transit to state {\tt tower} was from state {\tt tower builder}. However, according to Algorithm Byz-Known-Size, agent $A$ would have entered state {\tt tower} with index $0$ in round $r+1$ and would have stayed in state {\tt tower} until round $t$ in which it has index $j$. Hence, $r+1=t-j$ and agent $A$ would have been in state {\tt tower} since round $t-j$, which would be a contradiction with our starting assumption.

Denote by $s\leq t$ the last round when agent $A$ enters state {\tt tower} from state {\tt wait-for-a-tower}. Without loss of generality, let us assume, for each good agent in $T_j$, the last time $s'\leq t$ it enters state {\tt tower} is such that $s\leq s'\leq t$. According to Algorithm Byz-Known-Size, agent $A$ is in state {\tt tower} and has index $j-(t-s)$ in round $s$. Moreover, in round $s-1$ agent $A$ decides to transit from state {\tt wait-for-a-tower} to state {\tt tower} because it is at a node in which there are at least $f+1$ agents in state {\tt tower} having index $j-(t-s)-1$: among them there is necessarily a good agent $B$ as there are at most $f$ Byzantine agents.  Still according to Algorithm Byz-Known-Size, agents $A$ and $B$ make together the $j-(t-s)$-th edge traversal of $EXPLO(n)$ and are in state {\tt tower} with index $j-(t-s)$ at the same node $v$ in round $s$. In round $s$, there are at least $f+1$ agents in state {\tt tower} having index $j-(t-s)$, otherwise that would imply either there is no tower located at node $v$ in round $t$ if $t=s$, or agent $A$ is in state {\tt failure} in round $s+1$ if $s<t$: in the former case we would get a contradiction with the existence of tower $T_j$, and in the latter case we would get a contradiction with the definition of round $s$. Thus, if $s<t$ then we know from Algorithm Byz-Known-Size that agents $A$ and $B$ make together the $j-(t-s)+1$-th edge traversal of $EXPLO(n)$ and are in state {\tt tower} with index $j-(t-s)+1$ at the same node $u$ in round $s+1$. Following a similar reasoning to that used above, we can prove that there are at least $f+1$ agents in state {\tt tower} having index $j-(t-s)+1$ in round $s+1$ at node $u$. So, if $s+1<t$, then we know from Algorithm Byz-Known-Size that agents $A$ and $B$ make together the $j-(t-s)+2$-th edge traversal of $EXPLO(n)$ and are in state {\tt tower} with index $j-(t-s)+2$ at the same node $w$ in round $s+2$. By induction, we can then prove that agents $A$ and $B$ are together in round $t$ and that agent $B$ belongs to $T_j$. However, in round $t$ we know that agent $B$ has been in state {\tt tower} without interruption since at least round $s-1$, which contradicts the fact that for each good agent in $T_j$ the last time $s'\leq t$ it enters state {\tt tower} is such that $s\leq s'\leq t$. Hence the last time agent $A$ transited to state {\tt tower} before round $t$ was neither from state {\tt wait-for-a-tower} nor from state {\tt tower builder}. Since an agent can transit to state {\tt tower} only from states {\tt wait-for-a-tower} or {\tt tower builder}, we get a contradiction with the existence of $A$, and thus the lemma holds.
\end{proof}

\begin{lemma}
\label{lem:lem3}
Let $A$ be a good agent that is either {\tt orange}, or {\tt red} or in state {\tt tower} in round $r$. Let $\rho_i$ be the configuration tested by agent $A$ in round $r$. All good agents, which do not declare gathering is over before round $r$, test $\rho_i$ in round $r$.
\end{lemma}

\begin{proof}
Assume by contradiction there exists some round $x$ (playing the role of round $r$ in the statement of the lemma) such that the lemma does not hold. Without loss of generality, let us assume that $x$ is the first round for which this lemma is false.

Let $B$ be a good agent that does not declare gathering is over before round $x$.

Since agent $A$ is either {\tt orange}, or {\tt red} or in state {\tt tower}, in round $x$ it has already finished part~1 of Algorithm Byz-Known-Size which consists in executing procedure $EXPLO(n)$ and then traversing all edges traversed in $EXPLO(n)$ in the reverse order and the reverse direction. In view of the fact that procedure $EXPLO(n)$ allows an agent to visit all nodes of the graph, we know that agent $B$ is not dormant in round $x$.

Since agent $B$ is not dormant in round $x$, to prove the lemma it is then enough to show that in round $x$ it cannot be in part~1 of Algorithm Byz-Known-Size or test a configuration $\rho_j$ such that $j<i$ or $j>i$.

First assume by contradiction that agent $B$ tests a configuration $\rho_j$ in round $x$ such that $j<i$. Let us consider the case when agent $A$ is {\tt orange} in round $x$. Since agents $A$ and $B$ do not declare that gathering is over before round $x$, Lemma~\ref{lem:lem1} implies that in round $x$ agent $B$ has spent at most $2T(EXPLO(n))+j*Q(n)\leq 2T(EXPLO(n))+(i-1)Q(n)$ rounds since its wake-up. In view of Lemma~\ref{lem:lem1} and the period of $T(EXPLO(n))+n$ rounds which is necessary to wait in substate {\tt yellow} to enter substate {\tt orange}, 
when in round $x$ agent $A$ has spent at least $2T(EXPLO(n))+(i-1)*Q(n)+T(EXPLO(n))+n$ rounds since its wake up. This implies that the delay between the starting rounds of agents $A$ and $B$ is greater than $T(EXPLO(n))$, which contradicts Proposition~\ref{prop:prop1}. 

Let us consider the case when agent $A$ is in state {\tt tower} in round $x$. If agent $A$ is {\tt red} or in state {\tt tower} in round $x-1$, then according to Algorithm Byz-Known-Size it also tests configuration $\rho_i$ in round $x-1$. However, since agent $B$ tests $\rho_j$ in round $x$, according to Algorithm Byz-Known-Size it is either in part~1 of the algorithm, or also in phase $j$, or in phase $j-1$ (only if $j>1$). Hence if agent $A$ is {\tt red} or in state {\tt tower} in round $x-1$, then the lemma does not hold in round $x-1$ and we get a contradiction with the fact that $x$ is the first round for which this lemma is false. So, assume that agent $A$ is neither {\tt red} nor in state {\tt tower} in round $x-1$. According to Algorithm Byz-Known-Size it is in state {\tt wait-for-a-tower} in round $x-1$. The reason for which it decides to transit to state {\tt tower} in round $x-1$ is due to the fact that there is a tower $\mathcal{T}_s$ in round $x-1$ for some index $s$. As for round $x$, agent $A$ also tests configuration $\rho_i$ in round $x-1$ in view of Algorithm Byz-Known-Size. By assumption, the lemma holds in round $x-1$, and thus there is at least a good agent in state {\tt tower} belonging to $\mathcal{T}_s$ in round $x-1$ which tests configuration $\rho_i$. However, as mentioned above agent $B$ cannot test configuration $\rho_i$ in round $x-1$. As a consequence, the lemma does not hold in round $x-1$ and we get a contradiction with fact that $x$ is the first round for which this lemma is false.

Let us now consider the case when agent $A$ is  {\tt red} in round $x$. According to Algorithm Byz-Known-Size it also tests configuration $\rho_i$ in round $x-1$ and there is a good agent $C$ (not necessarily different from $A$) that is {\tt orange} or {\tt red}. Since the lemma holds in round $x-1$, agent $C$ also tests $\rho_i$ in round $x-1$. From this point, using as above the fact that agent $B$ does not test configuration $\rho_i$ in round $x-1$, we know that the lemma does not hold in round $x-1$ and we obtain the same contradiction.

Hence agent $B$ does not test a configuration $\rho_j$ in round $x$ such that $j<i$. In a similar way, we can prove that agent $A$ cannot be in part~1 of Algorithm Byz-Known-Size in round $x$. So, it remains to prove that agent $A$ does not test a configuration $\rho_j$ in round $x$ such that $j>i$. 

Assume by contradiction agent $B$ tests a configuration $\rho_j$ in round $x$ such that $j>i$. According to Algorithm Byz-Known-Size and Proposition~\ref{prop:prop2}, if agent $A$ is {\tt orange} (resp. {\tt red}) in round $x$, we know that it has spent at most $2T(EXPLO(n))+3n$ (resp. $3T(EXPLO(n))+4n$) rounds in phase $i$ when in round $x$. Thus, in round $x$, if agent $A$ is {\tt orange} (resp. {\tt red}), agent $A$ has spent at most $2T(EXPLO(n))+(i-1)*Q(n)+3T(EXPLO(n))+4n$ rounds since its wake up in view of Lemma~\ref{lem:lem1}. However according to Lemma~\ref{lem:lem1}, in round $x$ agent $B$ has spent at least $2T(EXPLO(n))+i*Q(n)=2T(EXPLO(n))+(i-1)*Q(n)+(10T(EXPLO(n))+9n)$ rounds since its wake up. Hence, the delay between the starting rounds of agents $A$ and $B$ is greater than $T(EXPLO(n))$, which contradicts Proposition~\ref{prop:prop1}. Thus agent $A$ is necessary in state {\tt tower} of phase $i$ in round $x$. According to Algorithm Byz-Known-Size, this implies agent $A$ is {\tt red} in round $x-1$ or there is a tower $\mathcal{T}_j$ in round $x-1$ for some $j\leq T(EXPLO(n))-1$. In the first case, we can show as above that the delay between the starting rounds of agents $A$ and $B$ is greater than $T(EXPLO(n))$ in round $x-1$ which is a contradiction with Proposition~\ref{prop:prop1}. So let us focus on the second case in which there is a tower $\mathcal{T}_j$ in round $x-1$. According to Algorithm Byz-Known-Size, agent $A$ tests configuration $\rho_i$ in round $x-1$. Since the lemma holds in round $x-1$, we then know that all good agents of tower $\mathcal{T}_j$ also test configuration $\rho_i$ in round $x-1$: among them there is at least one good agent $C$ that enters state tower of phase $i$ from substate {\tt red} of phase $i$ in some round $x'\leq x-1$ in view of Lemma~\ref{lem:lem2}. Hence in round $x-1$ agent $C$ has spent at most $2T(EXPLO(n))+(i-1)*Q(n)+3T(EXPLO(n))+4n+j\leq 2T(EXPLO(n))+(i-1)*Q(n)+4T(EXPLO(n))+4n$ rounds since its wake up in view of Lemma~\ref{lem:lem1}, while agent $B$ has spent at least $2T(EXPLO(n))+(i-1)*Q(n)+(10T(EXPLO(n))+9n)-1$. As before, we get a contradiction with Proposition~\ref{prop:prop1}. Thus, agent $B$ does not test a configuration $\rho_j$ in round $x$ such that $j>i$.

We then get a contradiction in all cases, which proves the lemma.

\end{proof}

\begin{lemma}
\label{lem:lem5}
Let $A$ be a good agent entering state {\tt tower} from state {\tt tower builder} in round $r$ at node $v$. Let $B$ be a good agent that does not declare that gathering is over before round $r$ and that does not enter state {\tt tower} from state {\tt tower builder} in round $r$ at node $v$. Agent $B$ is in state {\tt wait-for-a-tower} in round $r$.
\end{lemma}

\begin{proof}
Assume by contradiction that there exists some round $x$ (playing the role of round $r$ in the statement of the lemma) such that the lemma does not hold. Without loss of generality, let us assume that $x$ is the first round for which this lemma is false. Let $\rho_i$ be the configuration that is tested by agent $A$ in round $x$. Let $k$ be the number of labels in configuration $\rho_i$ (recall that all the labels in configuration $\rho_i$ are distinct as they are meant to represent only the labels of good agents). According to Lemma~\ref{lem:lem3}, we know that agent $B$ also tests configuration $\rho_i$ in round $x$. To get a contradiction, we will prove that agent $B$ cannot be in any state of phase $i$ in round $x$ except state {\tt wait-for-a-tower}.

Agent $B$ cannot be in state {\tt setup} of phase $i$ in round $x$. Indeed, since agents $A$ and $B$ do not declare that gathering is over before round $x$, in view of Lemma~\ref{lem:lem1} and Proposition~\ref{prop:prop1}, we know that the delay between the two rounds in which they start testing configuration $\rho_i$ is at most $T(EXPLO(n))$. However by Proposition~\ref{prop:prop2}, agent $B$ cannot spend more than $n$ rounds in state {\tt setup} of phase $i$, while agent $A$ needs to spend at least $T(EXPLO(n))+n+2$ rounds in phase $i$ before entering state {\tt tower} of phase $i$. Hence agent $B$ cannot be in state {\tt setup} in round $x$, otherwise we would get a contradiction with the fact that the delay between the two rounds in which they start testing configuration $\rho_i$ is at most $T(EXPLO(n))$. 

In round $x$ agent $B$ cannot be in state {\tt tower} of phase $i$. To show this, we first prove that in round $x$ agent $B$ cannot be in state {\tt tower} of phase $i$ with index $0$ (we will then prove it cannot be in this state even with any positive index). If in round $x$ agent $B$ is in state {\tt tower} of phase $i$ with index $0$, then according to Algorithm Byz-Known-Size agents $A$ and $B$ are {\tt red} in round $x-1$ and enter state {\tt tower} in round $x$. In round $x-1$, agents $A$ and $B$ are not at the same node, otherwise according to Algorithm Byz-Known-Size they are still together in round $x$ and we then get a contradiction with the assumption that agent $B$ does not enter state {\tt tower} from state {\tt tower builder} in the same round and at the same node as agent $A$. Hence in round $x-1$ agents $A$ and $B$ are {\tt red} at distinct nodes and decide to transit to state {\tt tower}: according to Lemma~\ref{lem:lem3}, they do so by testing the same configuration $\rho_i$.
According to Algorithm Byz-Known-Size, there are then at least $2k$ {\tt red} agents in round $x-1$: a group of at least $k$ at the node where agent $A$ is located, and another group of at least $k$ at the node of agent $B$. Since  $k\geq f+1$, there is at least $k+1$ good agents that are {\tt red} in round $x-1$. Thus, according to Lemma~\ref{lem:lem3}, there are at least $k+1$ good agents that are {\tt red} in round $x-1$ which test the same configuration $\rho_i$. However, this implies that each of these good {\tt red} agents has its label in configuration $\rho_i$ (otherwise it would have been impossible for at least one of them to transit to state {\tt tower builder} of phase $i$ according to the rules of state {\tt setup}): there are then at least $k+1$ distinct labels in $\rho_i$ which contradicts the definition of $k$. Let us now prove that in round $x$ agent $B$ cannot be in state {\tt tower} and have an index $j>0$. If agent $B$ is in state {\tt tower} of phase $i$ and has index $j>0$ in round $x$, then, in view of Lemma~\ref{lem:lem2}, there is a good agent $C$ in state {\tt tower} having index $j-1$ in round $x-1$ which has been in state {\tt tower} since round $x-j$. Hence, in round $x-j$ agent $C$ would be in state {\tt tower} with index $0$, which means that agent $C$ enters 
state {\tt tower} from state {\tt tower builder} in round $x-j$ according to Algorithm Byz-Known-Size. Since index $j$ cannot be greater than $T(EXPLO(n))$, we know from Algorithm Byz-Known-Size that agent $A$ is {\tt red} in round $x-j$ as it necessarily spent $T(EXPLO(n))+n$ rounds as a {\tt red} agents before entering state {\tt tower} in round $x$. Hence in round $x-j$, there is an agent $C$ that enters state {\tt tower} while agent $A$ is {\tt red} in state {\tt tower builder}: the lemma does not hold in round $x-j$, which contradicts the fact that $x$ is the first round when the lemma is false.

Now that the cases {\tt setup} and {\tt tower} have been excluded, it remains to prove that in phase $i$, agent $B$ cannot be either in state {\tt tower builder} or in state {\tt failure}.

Let us first prove that agent $B$ cannot be in state {\tt tower builder} of phase $i$ in round $x$. To prove this, it is enough to show that it cannot be either {\tt yellow}, or {\tt orange}, or {\tt red}. 

If in round $x$ agent $B$ is {\tt yellow} and tests configuration $\rho_i$, it cannot be {\tt yellow} and test configuration $\rho_i$ in round $x-T(EXPLO(n))-n-1$. Hence when in round $x-T(EXPLO(n))-n-1$, agent $B$ has spent at most $2T(EXPLO(n))+(i-1)*Q(n)+n$ rounds since its wake up in view of Algorithm Byz-Known-Size, Proposition~\ref{prop:prop2} and Lemma~\ref{lem:lem1}. However, according to Algorithm Byz-Known-Size, when in round $x-T(EXPLO(n))-n-1$, agent $A$ is {\tt yellow} or {\tt orange}, tests configuration $\rho_i$ and sees at least $k$ {\tt orange} agents at its current node. From Lemma~\ref{lem:lem3} and the fact that $k\geq f+1$, there is a good agent $C$ that is {\tt orange} and tests configuration 
$\rho_i$ in round $x-T(EXPLO(n))-n-1$. When in this round, by Lemma~\ref{lem:lem1} and Proposition~\ref{prop:prop2}, agent $C$ has spent at least $3T(EXPLO(n))+(i-1)*Q(n)+n+1$ rounds since its wake up (i.e. at least $T(EXPLO(n))+1$ rounds more than agent $B$), which contradicts Proposition~\ref{prop:prop1}.

In round $x$ agent $B$ cannot be {\tt orange}. Indeed, if agent $B$ is {\tt orange} in round $x$, then according to Algorithm Byz-Known-Size and Lemma~\ref{lem:lem3} we have the following fact: there are at least $k$ agents such that each of them is either {\tt yellow} or {\tt orange} in round $x-1$ which test the same configuration $\rho_i$, and there are at least $k$ {\tt red} agents in round $x-1$ which test the same configuration $\rho_i$. Since  $k\geq f+1$, there are at least $k+1$ good agents that are in state {\tt tower builder} in round $x-1$ which test the same configuration $\rho_i$. However, this implies that each of these good agents has its label in configuration $\rho_i$ (otherwise it is impossible to transit to state {\tt tower builder} of phase $i$ according to the rules of state {\tt setup}): there are then at least $k+1$ distinct labels in $\rho_i$ which contradicts the definition of $k$. In round $x$, agent $B$ cannot be {\tt red}. Indeed, if agent $B$ is {\tt red}
in round $x$ that means there is a round $x-p$ ($p\leq T(EXPLO(n))+n$) in which it is in a group of at least $k$ agents having a color belonging to $\{yellow, orange\}$ while agent $A$ is in a group of at least $k$ {\tt red} agents. From this point, using similar arguments to those used just above to prove agent $B$ is not {\tt orange} in round $x$, we get a contradiction with the definition of $k$.

To end the proof, it remains to show that agent $B$ is not in state {\tt failure} of phase $i$ in round $x$. Assume by contradiction, it is in state {\tt failure} of phase $i$ in round $x$. Let $x'<x$ be the last time agent $B$ decided to transit to state {\tt failure} (recall that when an agent decides to transit from some state X to some state Y in some round $t$, the agent remains in state X in round $t$ and is in state Y in round $t+1$). In round $x'$, agent $B$ is either in state {\tt tower} of phase $i$ or in state {\tt wait-for-a-tower} of phase $i$. If agent $B$ is in state {\tt wait-for-a-tower} of phase $i$ and decides to transit to state {\tt failure} in round $x'$, then by Lemma~\ref{lem:lem1}, Proposition~\ref{prop:prop2} and Algorithm Byz-Known-Size, when in round $x$ agent $B$ has spent at least $2T(EXPLO(n))+(i-1)Q(n)+5T(EXPLO(n))+4n$ rounds since its wake up. However still by Lemma~\ref{lem:lem1}, Proposition~\ref{prop:prop2} and Algorithm Byz-Known-Size, when in round $x$ agent $A$ has spent at most $2T(EXPLO(n))+(i-1)Q(n)+3T(EXPLO(n))+4n+1$ rounds since its wake up, which contradicts Proposition~\ref{prop:prop1}. If agent $B$ is in state {\tt tower} of phase $i$ with some index $j\geq 0$ in round $x'$, then in view of Lemmas~\ref{lem:lem2} and~\ref{lem:lem3} as well as Algorithm Byz-Known-Size, we know there is a good agent $C$ (not necessarily different from agent $B$) that enters state {\tt tower} of phase $i$ from state {\tt tower builder} of phase $i$ in round $x'-j$. Hence in view of Lemma~\ref{lem:lem3} and the fact the lemma holds in all rounds prior to round $x$, in round $x'-j$ agent $A$ either also enters state {\tt tower} of phase $i$ from state {\tt tower builder} of phase $i$, or is in state {wait-for-a-tower} of phase $i$. However in both these cases, it is then impossible for agent $A$ to enter state {\tt tower} of phase $i$ from state {\tt tower builder} of phase $i$ in round $x$ according to Algorithm Byz-Known-Size, which is a contradiction.

Thus, agent $B$ can be only in state {\tt wait-for-a-tower} of phase $i$ in round $x$, which proves the lemma.

\end{proof}

\begin{lemma}
\label{lem:lem4}
In any round there is at most one tower.
\end{lemma}

\begin{proof}
Assume by contradiction there exists some round $r$ when there are two distinct towers $\mathcal{D}_j$ and $\mathcal{T}_k$. The members of $\mathcal{D}_j$ (resp. $\mathcal{T}_k$) are in state {\tt tower} and all have index $j$ (resp. index $k$). Note that in the case where $\mathcal{D}_j$ and $\mathcal{T}_k$ are at the same node, indexes $j$ and $k$ are different from each other, otherwise we would have $\mathcal{D}_j=\mathcal{T}_k$ according to the definition of a tower. In the other case, index $j$ is not necessarily different to index $k$. Without loss of generality, we assume in the rest of this proof that $j\geq k$. By Lemma~\ref{lem:lem2}, there is a good agent $A\in\mathcal{D}_j$ (resp. $B\in\mathcal{T}_k$) which has been in state {\tt tower} since round $r-j$ (resp. round $r-k$). Let $\rho_i$ be the configuration tested by agent $A$ in round $r$. By Lemma~\ref{lem:lem3} agent $B$ also tests configuration $\rho_i$ in round $r$. Hence, in view of the fact that 
agent $A$ (resp. agent $B$) is in state {\tt tower} of phase $i$ with index $j$ (resp. index $k$) in round $r$, agent $A$ (resp. agent $B$) enters state {\tt tower} of phase $i$ in round $r-j$ (resp. round $r-k$) from state {\tt tower builder} of phase $i$. Denote by $u$ the node occupied by agent $A$ in round $r-j$. By Lemmas~\ref{lem:lem3} and~\ref{lem:lem5}, in round $r-j$ agent $B$ is either in state {\tt wait-for-a-tower} of phase $i$ or also enters state {\tt tower} of phase $i$ from state {\tt tower builder} at node $u$.  The first case implies that agent $B$ cannot enter state {\tt tower} of phase $i$ from state {\tt wait-for-a-tower} of phase $i$ in round $r-k$, which is a contradiction. The second case implies that agents $A$ and $B$ enter together state {\tt tower} of phase $i$ from state state {\tt tower builder} in the same round $r-j=r-k$ and at the same node, and thus belong to the same tower in round $r$ according to Algorithm Byz-Known-Size, which is also a contradiction and thus proves the lemma.

\end{proof}

We are now ready to prove Lemmas~\ref{lem:6} and~\ref{lem:7}.

\begin{lemma}
\label{lem:6}
If a good agent declares gathering is over at node $v$ in round $r$, then all good agents are at node $v$ in round $r$ and declare that gathering is over in round $r$.
\end{lemma}

\begin{proof}
Assume by contradiction there is a good agent $A$ that declares gathering is over at node $v$ in round $r$ but there is a good agent $B$ that does not make such a declaration at the same node and in the same round. Without loss of generality, we assume round $r$ is the first round when an agent declares gathering is over. Thus, agent $B$ does not declare gathering is over before round $r$.

According to Algorithm Byz-Known-Size, there is a tower $\mathcal{T}_{P(n)}$ in round $r$ at node $v$: all agents belonging to $\mathcal{T}_{P(n)}$ are in state {\tt tower} and have index $P(n)$ (which corresponds to the number of edge traversals of procedure $EXPLO(n)$). By Lemma~\ref{lem:lem2}, there is a good agent $C$ that proceeded to an entire execution of $EXPLO(n)$ from round $r-P(n)$ to round $r$. None of agents in $\mathcal{T}_{P(n)}$ could be agent $B$ because all good agents belonging to $\mathcal{T}_{P(n)}$ in round $r$ declare gathering is over according to the rules of state {\tt tower}. In particular agent $C$ cannot be agent $B$. When agent $C$ enters state  {\tt tower} in round $r-P(n)$, agent $B$ is in state {\tt wait-for-a-tower}, otherwise in view of Lemma~\ref{lem:lem5} and Algorithm Byz-Known-Size, agents $C$ and $B$ both belong to $\mathcal{T}_{P(n)}$, which is a contradiction.

If agent $B$ does not enter state {\tt failure} in some round of $\{r-P(n)+1,\cdots,r\}$, then either agent $B$ remains in state {\tt wait-for-a-tower} from round $r-P(n)$ to round $r$, or agent $B$ enters state {\tt tower} in some round $x\in \{r-P(n)+1,\cdots,r\}$. In the former case, that means agent $B$ is at node $v$ in round $r$. Indeed, otherwise agent $B$ shares with agent $C$ the same node (because $C$ makes an entire traversal as an agent in state {\tt tower}) in some round $x\in \{r-P(n),\cdots,r-1\}$, and then transits to state {\tt tower} in round $x+1\leq r$ (because agent $C$ is always in a tower during its entire traversal), which is a contradiction. However due to the presence of a tower in round $r$ at node $v$, agent $B$ declares gathering is over in the same round and at the same node according to the rules of state {\tt wait-for-a-tower}, which is again a contradiction. In the latter case, agent $B$ enters state {\tt tower} in some round $x\in \{r-P(n)+1,\cdots,r\}$. By Lemma~\ref{lem:lem4} and the fact that agent $C$ is always in state {\tt tower} from round $r-P(n)$ to $r$, agent $B$ belongs to the same tower as agent $C$ from rounds $x$ to $r$. Thus agent $B$ is included in tower $\mathcal{T}_{P(n)}$ in round $r$, which is a contradiction.

Hence agent $B$ enters state {\tt failure} in some round $x\in\{r-P(n)+1,\cdots,r\}$. By Lemma~\ref{lem:lem3}, agents $A$, $B$ and $C$ test the same configuration $\rho_i$ in round $r$ for some positive integer $i$. According to Algorithm Byz-Known-Size, we know agent $C$ also tests configuration $\rho_i$ in round $x$. Since agent $C$ is in state {\tt tower} in round $x$, agent $B$ also tests configuration $\rho_i$ in round $x$, according to Lemma~\ref{lem:lem3}. Hence in view of the fact that agent $B$ enters state {\tt failure} by testing configuration $\rho_i$ in round $x$ and the fact that it still tests configuration $\rho_i$ in round $r$, we know that agent $B$ is in state {\tt failure} of phase $i$ in round $r$. However by Lemma~\ref{lem:lem1}, Proposition~\ref{prop:prop2} and Algorithm Byz-Known-Size, when in round $r$, agent $B$ has spent at least $2T(EXPLO(n))+(i-1)*Q(n)+5T(EXPLO(n))+4n+1$ rounds since its wake up, while agent $C$ has spent at most $2T(EXPLO(n))+(i-1)*Q(n)+4T(EXPLO(n))+4n$ rounds: we thus get a contradiction with Proposition~\ref{prop:prop1}. Hence agent $B$ does not exist and the lemma holds.

\end{proof}

\begin{lemma}
\label{lem:7}
There is at least one good agent that ends up declaring that gathering is over.
\end{lemma}

\begin{proof}
Assume by contradiction no agent ends up declaring gathering is over. Let $\rho_i$ be a configuration, belonging to enumeration $\Theta$, which corresponds to the initial configuration of all good agents in the graph. Configuration $\rho_i$ is said to be good. Let $k$ be the number of labels in configuration $\rho_i$: $k\geq f+1$ because every tested configuration contains at least $f+1$ labels. Since configuration $\rho_i$ is good, in view of Proposition~\ref{prop:prop1}, Lemma~\ref{lem:lem1} and the fact that at least one good agent is woken up by the adversary, every good agent reaches state {\tt tower builder} of phase $i$ possibly in different rounds but at the same node $v$ corresponding to the node where the agent having the smallest label is initially located. Denote by $A$ the first agent to enter state {\tt tower builder} of phase $i$ in some round $r$. Agent $A$ is thus {\tt yellow} in round $r$ at node $v$. By Propositions~\ref{prop:prop1} and~\ref{prop:prop2}, Lemma~\ref{lem:lem1} and the fact that configuration $\rho_i$ is good, we know that each of all good agents enters substate {\tt yellow} of phase $i$ at node $v$ in some round $x\in \{r,\cdots,r+T(EXPLO(n))+n-1\}$. Note that if a good agent, which is in substate {\tt yellow} of phase $i$ in round $x'\in \{r,\cdots,r+T(EXPLO(n))+n-2\}$ at node $v$, is no longer in substate {\tt yellow} of phase $i$ in round $x'+1$ at node $v$, this implies that the good agent sees at least $k$ {\tt orange} agents at node $v$ in round $x'$. Among this group of at least $k$ {\tt orange} agents, at least one agent, call it $C$, is good as $k\geq f+1$. By Lemma~\ref{lem:lem3}, agent $C$ also tests configuration $\rho_i$ in round $x'$, and thus agent $C$ entered state {\tt tower builder} of phase $i$ in round $x'-T(EXPLO(n))-n-1<r$ at the latest, which contradicts the definition of agent $A$. Hence, each of all good agents enters substate {\tt yellow} of phase $i$ at node $v$ in some round $x\in \{r,\cdots,r+T(EXPLO(n))+n-1\}$ and remains in substate {\tt yellow} of phase $i$ at node $v$ at least until round $r+T(EXPLO(n))+n-1$ included. We have therefore the following claim.


{\bf Claim~1.} All good agents are {\tt yellow} at node $v$ and test configuration $\rho_i$ in round $r+T(EXPLO(n))+n-1$.

From Claim~1 and Algorithm Byz-Known-Size, the $k$ good agents become {\tt red} together before round $r+2T(EXPLO(n))+2n$. From this point, we know the $k$ good agents will then enter state {\tt tower} of phase $i$ in the same round at node $v$ and make together an entire execution of procedure $EXPLO(n)$: according to the rules of state {\tt tower}, agent $A$ declares gathering is over at the end of this execution. So, we get a contradiction, and the lemma holds.
\end{proof}

From Lemmas~\ref{lem:6} and~\ref{lem:7}, we know that Algorithm Byz-Known-Size and Theorem~\ref{theo:theo3} are valid.

\section{Unknown graph size}
\label{sec:unknown}

In this section, we consider the same problem, except we assume that the agents are not initially given the size of the graph. Under this harder scenario, we aim at proving the following theorem.
\vspace{-0.3cm}
\begin{theorem}
\label{theo:theo4}
Deterministic $f$-Byzantine gathering of $k$ good agents is possible in any graph of unknown size if, and only if $k\geq f+2$.
\end{theorem}
\vspace{-0.3cm}

As mentioned in Subsection~\ref{sub:subm}, we know from \cite{DieudonnePP14} that:
\vspace{-0.3cm}
\begin{theorem}[\cite{DieudonnePP14}]
\label{theo:theo5}
Deterministic $f$-Byzantine gathering of $k$ good agents is not possible in some graphs of unknown size if $k\leq f+1$.
\end{theorem}
\vspace{-0.3cm}
In view of Theorem~\ref{theo:theo5}, it is then enough to show the following theorem in order to prove Theorem~\ref{theo:theo4}.
\vspace{-0.3cm}
\begin{theorem}
\label{theo:theo6}
Deterministic $f$-Byzantine gathering of $k$ good agents is possible in any graph of unknown size if $k\geq f+2$.
\end{theorem}
\vspace{-0.3cm}
Hence, similarly as in Section~\ref{sec:sec1}, the rest of this section is devoted to showing a deterministic algorithm that gathers all good agents, but this time in an arbitrary network of unknown size and provided there are at least $f+2$ good agents.

Before giving the algorithm, which we call \emph{Algorithm Byz-Unknown-Size}, let us provide some intuitive ingredients on which our solution is based.

The algorithm of this section displays a number of similarities with the algorithm of the previous section, but there are also a number of changes to tackle the  non-knowledge of the network size. Among the most notable changes, there is firstly the way of enumerating the configurations. Previously, the agents were considering the enumeration $\Theta=(\rho_1,\rho_2,\rho_3,\cdots)$ of $\mathcal{P}$ where $\mathcal{P}$ is the set of every configuration corresponding to a $n$-node graph in which there are at least $f+1$ robots with pairwise distinct labels. Now, instead of considering $\Theta$, the agents will consider the enumeration $\Omega=(\phi_1,\phi_2,\phi_3,\cdots)$ of $\mathcal{Q}$ where $\mathcal{Q}$ is the set of all configurations corresponding to a graph of any size (instead of size $n$ only) in which there are at least $f+2$ agents (instead of at least $f+1$) with pairwise distinct labels. Note that, as for set $\mathcal{P}$, set $\mathcal{Q}$ is also recursively enumerable.

Another change stems from the function performed by a tower, which we also find here. In Algorithm Byz-Known-Size, the role of a tower was to fetch all awaiting good agents (which know that the tested configuration is not good) via procedure $EXPLO(n)$: in the new algorithm, we keep the exact same strategy. However, to be able to use procedure $EXPLO$ with a parameter corresponding to the size of the network, it is necessary, for the good agents that are members of a tower, to know this size. Hence, in our solution, before being considered as a tower and then authorized to make a traversal of the graph, a group of agents will have to learn the size of the graph. To do this, at least each good agent of the group will be required to make a simulation of procedure $EST$ by playing the role of an explorer and using the others as its token. To carry out these simulations, it is also required for the group of agents to contain initially at least $f+2$ members (explorer + token), even if subsequently it is required for a group of agents forming a tower to contain at least $f+1$ members. Our algorithm is designed in such a way that if during the simulation of procedure $EST$ by an agent playing the role of an explorer, we have the guarantee there are always at least $f+1$ agents playing the role of its token, then the explorer will be able to recognize its own token without any ambiguity (and thus will act as if it performed procedure $EST$ with a ``genuine'' token). Of course, the agents will not always have such a guarantee (especially due to the possible bad behavior of Byzantine agents when testing a wrong configuration) and will not be able to detect in advance whether they will have it or not. Besides, some other problems can arise including, for example, some Byzantine explorer which takes too much time to explore the graph (or worse still, ``never finishes'' the exploration). However we will show that in all cases, the good agents can never learn an erroneous size of the graph (even with the duplicity of Byzantine agents when testing a wrong configuration). We also show that good agents are assured of learning the size of the network when testing a good configuration at the latest (as the creation of a group of at least $f+2$ agents and the aforementioned guarantee are ensured when testing a good configuration). As for Algorithm Byz-known-Size, in the worst case the good agents will have to wait until assuming a good hypothesis about the real initial configuration, in order to declare gathering is over. 

Despite the fact Algorithm Byz-Unknown-Size has also a number of technical changes compared with Algorithm Byz-Known-Size (e.g., the duration of waiting periods that are adjusted according to the new context), it should be noted that some parts of Algorithm Byz-Unknown-Size are inevitably almost identical to some of those of Algorithm Byz-known-Size  (this is particularly the case for the description of state {\tt tower}). However for easy readability, we made the choice of writing completely these parts instead of explaining this or that part ``is the same as in Algorithm Byz-Known-Size except that...''.

We now give a detailed description of the algorithm.

{\bf Algorithm Byz-Unknown-Size}

The agent works in phases numbered $1,2,3,\cdots$. During the execution of each phase, the agent can be in one of the following seven states: {\tt setup}, {\tt tower builder}, {\tt token}, {\tt explorer}, {\tt tower}, {\tt wait-for-a-tower}, {\tt failure}. Below we describe the actions of an agent $A$ having label $l_A$ in each of the states as well as the transitions between these states within phase $i$. As in Algorithm Byz-Known-Size, we assume that in every round agent $A$ tells the others (sharing the same node as agent $A$) in which state it is. At the beginning of phase $i$, agent $A$ enters state {\tt setup}. We denote by $Q$ the time spent by agent $A$ executing Algorithm Byz-Unknown-Size before starting phase $i$.

  \noindent
 {\bf State} {\tt setup}.

Let $\phi_i$ be the $i$-th configuration of enumeration $\Omega=(\phi_1,\phi_2,\phi_3,\cdots)$ (refer to its description given in Section~\ref{sec:unknown}). Let $n_i$ be the number of nodes in configuration $\phi_i$. Agent $A$ starts executing procedure $EXPLO(n_i)$. Once the execution of $EXPLO(n_i)$ is accomplished, the agent backtracks to its starting node by traversing all edges traversed in $EXPLO(n_i)$ in the reverse order and the reverse direction. When the backtrack is done, the agent continues with this state via the following rules.
If $l_A$ is not in $\phi_i$, then it transits to state {\tt wait-for-a-tower}. Otherwise, let $X$ be the set of the shortest paths in $\phi_i$ leading from the node containing the agent having label $l_A$, to the node containing the smallest label of the supposed configuration. Each of paths belonging to $X$ is represented as the corresponding sequence of port numbers. Let $\pi$ be the lexicographically smallest path in $X$ (the lexicographic order can be defined using the total order on the port numbers). Agent $A$ follows path $\pi$ in the real network. If, following path $\pi$, agent $A$ has to leave by a port number that does not exist in the node where it currently resides, then it transits to state {\tt wait-for-a-tower}. In the same way, it also transits to state {\tt wait-for-a-tower} if, following path $\pi$, agent $A$ enters at some point a node by a port number which is not the same as that of path $\pi$. Once path $\pi$ is entirely followed by agent $A$, it transits to state {\tt tower builder}.
  \vspace*{0.2cm}
 
  \noindent
 {\bf State} {\tt tower builder}.


When in state {\tt tower builder}, agent $A$ can be in one of the following three substates: {\tt yellow}, {\tt orange}, {\tt red}. In all of these substates the agent does not make any move: it stays at the same node denoted by $v$. At the beginning, agent $A$ enters substate {\tt yellow}. As for Algorithm Byz-Known-Size, we will sometimes use a slight misuse of language by saying an agent ``is {\tt yellow}" instead of ``is in substate {\tt yellow}". We will also use the same kind of shortcut for the two other colors. In addition to its state, we also assume that in every round agent $A$ tells the others in which substate it is.

{\bf Substate} {\tt yellow}

Let $k$ be the number of labels in configuration $\phi_i$.
Agent $A$ waits $T(EXPLO(n_i))+n_i + Q$ rounds. If during this waiting period, there are at some point at least $k$ {\tt orange} agents at node $v$ then agent $A$ transits to substate {\tt red}. Otherwise, if at the end of this waiting period there are not at least $k$ agents residing at node $v$ such that each of them is either {\tt yellow} or {\tt orange}, then agent $A$ transits to state {\tt wait-for-a-tower}, else it transits to substate {\tt orange}.
  
{\bf Substate} {\tt orange}

Agent $A$ waits at most $T(EXPLO(n_i))+n_i + Q$ rounds to see the occurrence of one of the following two events. The first event is that there are not at least $k$ agents residing at node $v$ such that each of them is either {\tt yellow} or {\tt orange}. The second event is that there are at least $k$ {\tt orange} agents residing at node $v$. Note that the two events cannot occur in the same round. If during this waiting period, the first (resp. second) event occurs then agent $A$ transits to state {\tt wait-for-a-tower} (resp. substate {\tt red}). If at the end of the waiting period, none of these events has occurred, then agent $A$ transits to substate {\tt wait-for-a-tower}. 

{\bf Substate} {\tt red}

Agent $A$ waits $T(EXPLO(n_i))+n_i + T(EST(n_i)) + Q$ rounds, and in every round of this waiting period it tells the others the phase number $i$. If there is a round during the waiting period in which there are not at least $k$ {\tt red} agents in phase $i$ at node $v$: agent $A$ then transits to state {\tt wait-for-a-tower} as soon as it notices this fact. Otherwise, at the end of the waiting period agent $A$ transits either to state {\tt explorer}, or to state {\tt token}, or to state {\tt failure} according to the following rule. Let $\mathcal{H}$ be the set of pairwise distinct labels such that $l\in \mathcal{H}$ iff there is at least one {\tt red} agent in phase $i$ having label $l$ at node $v$ in the last round of the waiting period. Let $|\mathcal{H}|$ be the cardinality of $\mathcal{H}$. If $|\mathcal{H}|>n_i$ then agent $A$ transits to state {\tt failure}. Otherwise $|\mathcal{H}|\leq n_i$ and agent $A$ applies the following instruction: if $l_A$ is the smallest label in $\mathcal{H}$ then agent $A$ transits to state {\tt explorer}, else it transits to state {\tt token}.

  \noindent
 {\bf State} {\tt explorer}.

We first briefly describe the procedure $EST$ based on \cite{ChalopinDK10} that will be subsequently adapted to our needs. In this procedure, the agent constructs a BFS tree rooted at node $r$
marked by the stationary token. In this tree it marks port numbers at all nodes. 
During the BFS traversal, some nodes are added to the BFS tree. {In the beginning, the agent adds the root $r$ and then it makes the {\em process} of $r$. The process of a node $w$ consists in checking all the neighbors of $w$ in order to determine whether 
some of them have to be added to the tree or not. When an agent starts the process of a node $w$, it goes to the neighbor reachable via port $0$ and then checks the neighbor.}

{When a neighbor $x$ of $w$ gets checked, the agent  
verifies if $x$ is equal to some node
previously added to the tree. To do this, for each node $u$ belonging to the current BFS tree, the agent travels from $x$ using the reversal $\overline{q}$ of the shortest path $q$ from $r$ to $u$ in the BFS tree (the path $q$ is
a sequence of port numbers). If at the end of this backtrack it meets the token, then $x=u$: in this case $x$ is not added to the tree as a neighbor of $w$ and is called $w$-{\em rejected}. If not, then $x\neq u$. Whether node $x$ is rejected or not, the agent then comes back to $x$ using the path $q$. If $x$ is different from all the nodes of the BFS tree, then it is added to the tree.}

{Once node $x$ is added to the tree or rejected, the agent makes an edge traversal in order to be located at $w$ and then goes to a non-checked neighbor of $w$, if any. The order, in which the neighbors of $w$ are checked, follows the increasing order of the port numbers of $w$.} 

{When all the neighbors of $w$ are checked, the agent proceeds as follows. Let $\mathcal{X}$ be the set of the shortest paths in the BFS tree leading from the root $r$ to a node $y$ having non-checked neighbors. If $\mathcal{X}$ is empty then procedure $EST$ is completed. Otherwise, the agent goes to the root $r$, using the shortest path from $w$ to $r$ in the BFS tree, and then goes to a node $y$ having non-checked neighbors, using the lexicographically smallest path from $X$. From there, the agent starts the process of $y$.}

We are now ready to give the description of state {\tt explorer}.

When entering this state, agent $A$ executes the procedure $EST'$ which corresponds to a simulation of procedure $EST$ with the following three changes. The first change concerns meetings with the token. Consider a verification if a node $x$, which is getting checked, is equal to some previously added node $u$. This verification consists in traveling from $x$ using the reverse path $\overline{q}$,  where $q$ is the path from the root $r$ to $u$ in the BFS tree and checking the presence of the token. If 
 at the end of the simulation of path $\overline{q}$ in $EST'$, agent $A$ is at a node containing at least $f+1$ agents in state {\tt token} of phase $i$, then it acts as if it saw the token in $EST$; otherwise it acts as if it did not see
 the token in $EST$. The second change occurs during the construction of the BFS tree: if at some point agent $A$ has added more than $n_i$ nodes in the BFS tree or has spent more than $T(EST(n_i))$ rounds executing the current simulation, then it drops the simulation and transits to state {\tt failure}. The third and final change occurs at the end of the simulation: if agent $A$ has added less than $n_i$ nodes in the BFS tree, then it transits to state {\tt failure}.

Once the execution of procedure $EST'$ is done, agent $A$ backtracks to the node where it was located when entering state {\tt explorer}. To do this, the agent traverses all edges traversed during the execution of procedure $EST'$ in the reverse order and the reverse direction. When the backtrack is done, agent $A$ has spent exactly $2T(EST(n_i))-t$ rounds in state explorer of phase $i$, for some integer $t\geq 0$. From this point, agent $A$ waits $t$ rounds. At the end of the waiting period, if agent $A$ does not share its current node with at least $f+1$ agents in state {\tt token} of phase $i$, then it transits to state {\tt failure}. Otherwise agent $A$ does share its current node with at least $f+1$ agents in state {\tt token} of phase $i$: in this case, if $l_A$ is the largest label in set $\mathcal{H}$ (this set was determined when agent $A$ was {\tt red} in this phase) then agent $A$ transits to state {\tt tower}, else it transits to state {\tt token}.

  \noindent
 {\bf State} {\tt token}.

While in this state, agent $A$ remains at the same node $v$, and in every round it tells the others the phase number $i$. Agent $A$ can transit to state {\tt token} either from state {\tt tower builder} or from state {\tt explorer}. Below, we distinguish both these cases. Let $j$ be the number of labels of $\mathcal{H}$ that are smaller than $l_A$.

\begin{itemize}
\item {Case~1: the last time agent $A$ transited to state {\tt token} was from state {\tt tower builder}.} In this case agent $A$ waits $2j.T(EST(n_i))$ rounds. If there is a round during the waiting period in which there are not at least $f+1$ agents in state {\tt token} of phase $i$ at node $v$: agent $A$ then transits to state {\tt failure} as soon as it notices this fact. Otherwise, at the end of the waiting period agent $A$ transits to state {\tt explorer}.

\item {Case~2: the last time agent $A$ transited to state {\tt token} was from state {\tt explorer}.} In this case agent $A$ waits $2(|\mathcal{H}|-j-1).T(EST(n_i))$ rounds. If there is a round during the waiting period in which there are not at least $f+1$ agents in state {\tt token} of phase $i$ at node $v$: agent $A$ then transits to state {\tt failure} as soon as it notices this fact. Otherwise, at the end of the waiting period agent $A$ transits to state {\tt tower}.

\end{itemize}

  \noindent
 {\bf State} {\tt tower}.

Agent $A$ can enter state {\tt tower} either from state {\tt token}, or state {\tt explorer} or state {\tt wait-for-a-tower}. While in this state, agent $A$ will execute all or part of procedure $EXPLO(n_i)$. In all cases we assume that, in every round, agent $A$ tells the others the edge traversal number of $EXPLO(n_i)$ it has just made (in addition to its state). We call this number the index of the agent. Below, we distinguish and detail the case when agent $A$ enters state {\tt tower} from state {\tt token} or {\tt explorer}, and the case when it enters state {\tt tower} from state {\tt wait-for-a-tower}.

When agent $A$ enters state {\tt tower} from state {\tt token} or {\tt explorer}, it starts executing procedure $EXPLO(n_i)$. In the first round, its index is $0$. Just after making the $j$-th edge traversal of $EXPLO(n_i)$, its index is $j$. Agent $A$ carries out the execution of $EXPLO(n_i)$ until its term, except if at some round of the execution the following condition is not satisfied, in which case agent $A$ transits to state {\tt failure}. Here is the condition: the node where agent $A$ is currently located contains a group $\mathcal{S}$ of at least $f+1$ agents in state {\tt tower} having the same index as agent $A$. $\mathcal{S}$ includes agent $A$ but every agent that is in the same node as agent $A$ is not necessarily in $\mathcal{S}$. If at some point this condition is satisfied and the index of agent $A$ is equal to $P(n_i)$, which is the total number of edge traversals in $EXPLO(n_i)$ (refer to Section~\ref{sec:pre}), then agent $A$ declares that gathering is over.

When agent $A$ enters state {\tt tower} from state {\tt wait-for-a-tower}, it has just made the $s$-th edge traversal of $EXPLO(n_i)$ for some $s$ (cf. state {\tt wait-for-a-tower}) and thus, its index is $s$. Agent $A$ executes the next edge traversals
 i.e., the $s+1$-th, $s+2$-th, $\cdots$, and then its index is successively $s+1$, $s+2$, etc. Agent $A$ carries out this execution until the end of procedure $EXPLO(n_i)$, except if the same condition as above is not fulfilled at some round of the execution of the procedure, in which case agent $A$ also transits to state {\tt failure}. As in the first case, if at some point the node where agent $A$ is currently located contains a group $\mathcal{S}$ of at least $f+1$ agents in state {\tt tower} having an index equal to $P(n_i)$, then agent $A$ declares that gathering is over.

\noindent
 {\bf State} {\tt wait-for-a-tower}.


Agent $A$ waits at most $7T(EXPLO(n_i))+4n_i+(2n_i+1)T(EST(n_i))+4Q$ rounds to see the occurrence of the following event: the node where it is currently located contains a group of at least $f+1$ agents in state {\tt tower} having the same index $t$.
If during this waiting period, agent $A$ sees such an event, we distinguish two cases. If $t<P(n_i)$, then it makes the $t+1$-th edge traversal of procedure $EXPLO(n_i)$ and transits to state {\tt tower}. If $t=P(n_i)$, then it declares that gathering is over.

Otherwise, at the end of the waiting period, agent $A$ has not seen such an event, and thus it transits to state {\tt failure}.

\noindent
 {\bf State} {\tt failure}.


Agent $A$ backtracks to the node where it was located at the beginning of phase $i$. To do this, agent $A$ traverses in the reverse order and the reverse direction all edges it has traversed in phase $i$ before entering state {\tt failure}. Once at its starting node, agent $A$ waits $16T(EXPLO(n_i))+9n_i+2(n_i+1)T(EST(n_i))+7Q-p$ rounds where $p$ is the number of elapsed rounds between the beginning of phase $i$ and the end of the backtrack it has just made. At the end of the waiting period, phase $i$ is over. In the next round, agent $A$ will start phase $i+1$.

\subsection{Proof of correctness}

The purpose of this section is to prove that Algorithm Byz-Unknown-Size is correct and that by extension Theorem~\ref{theo:theo6} holds.

For any positive integer $i$, we say that a good agent $A$ tests configuration $\phi_i$ when it executes the $i$-th phase of Algorithm Byz-Unknown-Size. We denote by $n_i$ the size of the graph in configuration $\phi_i$ and by $n$ the (unknown) size of the network where the agents currently evolve. We assume that the smallest integer $i$ such that $n_i\geq n$ is $\alpha$.

According to state failure, we have the following proposition.

\begin{proposition}
\label{prop:prop41}
At the beginning of every phase it executes, a good agent is at the node where it was woken up.
\end{proposition}

\begin{lemma}
\label{lem:lem41}
Let $A$ be a good agent which starts executing the $i$-th phase of Algorithm Byz-Unknown-Size in some round $r$. Let $Z(n_i)=16T(EXPLO(n_i))+9n_i+2(n_i+1)T(EST(n_i))+7Q$ where $Q$ is the number of rounds spent by agent $A$ executing the algorithm before round $r$. The following two properties hold.

\begin{itemize}

\item Property~1. Agent $A$ either spends exactly $Z(n_i)$ rounds executing the $i$-th phase or it will declare that gathering is over after having spent at most $Z(n_i)$ rounds in the $i$-th phase.

\item Property~2. Let $B$ be a good agent different from agent $A$ which starts executing the $i$-th phase of Algorithm Byz-Unknown-Size in some round $r'$ (round $r'$ is not necessarily different from round $r$). Agent $B$ has also spent exactly $Q$ rounds executing Algorithm Byz-Unknown-Size before starting phase $i$.
\end{itemize}
\end{lemma}

\begin{proof}
Using similar arguments to those used in the proof of Lemma~\ref{lem:lem1}, we can prove that the first property holds. Concerning the second property, it is a corollary of the first property.
\end{proof}

In view of Lemma~\ref{lem:lem41}, we know that for every integer $i$, the good agents that do not declare gathering is over before entering phase $i$, all spend the exact same time executing Algorithm Byz-Unknown-Size before entering phase $i$ (whether they enter it in the same round or not). In the rest of this section, we denote by $Q_j$ the number of rounds spent before entering phase $j$ by any agent that does not declare gathering is over before starting the $j$-th phase of Algorithm Byz-Unknown-Size.

\begin{proposition}
\label{prop:prop42}
Let $A$ and $B$ be two good agents such that agent $A$ is woken up by the time agent $B$ is woken up. If agent $A$ does not declare gathering is over before starting phase $\alpha$, then the delay between the starting rounds of agents $A$ and $B$ is at most $Q_\alpha+T(EXPLO(n_\alpha))$.
\end{proposition}  

\begin{proof}
If agent $A$ is woken up in some round, it ends up starting phase $\alpha$ after having spent exactly $Q_\alpha$ rounds. When entering phase $\alpha$, agent $A$ is in state {\tt setup} and first execute procedure $EXPLO(n_\alpha)$ according to Algorithm Byz-Unknown-Size. Since by definition $n_\alpha\geq n$, then the propositions holds. 
\end{proof}

Even if Algorithm Byz-Unknown-Size has several changes compared with Algorithm Byz-known-Size (in particular the two extra states {\tt explorer} and {\tt token} to take into account here), we can prove Lemmas~\ref{lem:lem42}, \ref{lem:lem43} and~\ref{lem:lem44} by using similar arguments to those used in the proofs of Lemmas~\ref{lem:lem2}, \ref{lem:lem3} and~\ref{lem:lem5}.

\begin{lemma}
\label{lem:lem42}
Let $T_j$ be a tower located at node $v$ in round $t$. There is at least one good agent in $T_j$ which has been in state {\tt tower} since round $t-j$.
\end{lemma}

\begin{lemma}
\label{lem:lem43}
Let $A$ be a good agent that is either {\tt orange}, or {\tt red} or in a state $\in$ \{{\tt tower}, {\tt token}, {\tt explorer}\} in round $r$. Let $\phi_i$ be the configuration tested by agent $A$ in round $r$. If $i\geq\alpha$,  then all good agents, which do not declare gathering is over before round $r$, test $\phi_i$ in round $r$.
\end{lemma}

\begin{lemma}
\label{lem:lem44}
Let $A$ be a good agent entering state {\tt explorer} or {\tt token} from state {\tt tower builder} in round $r$ at node $v$ by testing a configuration $\phi_i$ such that $i\geq\alpha$. Let $B$ be a good agent that does not declare that gathering is over before round $r$ and that does not enter state {\tt explorer} or {\tt token} from state {\tt tower builder} in round $r$ at node $v$. Agent $B$ is in state {\tt wait-for-a-tower} in round $r$.
\end{lemma}

\begin{lemma}
\label{lem:lem45}
Let $A$ be a good agent that tests configuration $\phi_i$ in round $r$. If agent $A$ enters state {\tt tower} from state {\tt token} or {\tt explorer} in round $r$ then $n_i=n$.
\end{lemma}

\begin{proof}

Let us first consider the case where agent $A$ enters state {\tt tower} from state {\tt explorer} in round $r$. According to Algorithm Byz-Unknown-Size agent $A$ computed a BFS tree $T$ while in state {\tt explorer} (just before transiting to state {\tt tower} of phase $i$). Before going any further, we need to prove the following claim.

{\bf Claim~1.} The size $m$ of $T$ is such that $m=n$.

If $n_i=n$ then the claim is true because according to Algorithm Byz-Unknown-Size agent $A$ can transit from state {\tt explorer} to state {\tt tower} in round $r$ only if $m=n_i$. So let us focus on the case where $n_i\ne n$. Denote by $u$ the node of the graph corresponding to the root of $T$. Denote by $T'$ the BFS tree rooted at a node corresponding to node $u$ and that would result from the execution of procedure $EST$ by an explorer having its own token at node $u$ which cannot disappear and which cannot be confused with another token. The size of $T'$ is therefore equal to $n$. 
If $T$ is identical to $T'$, then $m=n$. However note that since $n_i\ne n$, we have $m\ne n_i$ and thus agent $A$ cannot transit from state {\tt explorer} to state {\tt tower} in round $r$ according to Algorithm Byz-Unknown-Size which is a contradiction with the definition of $r$. 

Therefore $T$ is necessarily different from $T'$. There are only two possible incidents that can lead to such a situation, according to the definition of $T$. The first one is that at some point during the exploration of agent $A$, its token vanished i.e, there is a round during the exploration when there are not $f+1$ agents at node $u$ that claim being in state {\tt token} of phase $i$. The other one is that at some point agent $A$ confused its token with another token i.e., it encountered during its exploration a group of at least $f+1$ agents at a node $v\ne u$ that claimed being in state {\tt token} of phase $i$. However note that any execution of procedure $EST'$ consists of alternating periods of two different types. The first type corresponds to periods when the agent processes a node and the second type corresponds to those when the agent moves to the next node to process it. During the periods of the second type, an agent does not use any token to move: it follows the same path regardless of whether it meets some token or not on its path. Hence, each of the two possible incidents describe above can only have an impact on the BFS tree only if they occur during a period of the first type when verifying whether a node has to be rejected or not. So denote by $t$ the first round in the construction of $T$ via procedure $EST'$ when agent $A$ adds a node to its BFS tree $T$ under construction that has to be rejected, or when agent $A$ rejects a node that has to be added to its BFS tree $T$ under construction. This round necessarily exists as otherwise $T'$ could not be different from $T$ according to the above explanations. We consider the only two possible cases (each of them leading to a contradiction). Let $k$ be the number of labels in configuration $\rho_i$.

\begin{itemize}

\item{Case~1: in round $t$ agent $A$ rejects a node  $x$ that has to be added to its BFS tree $T$ under construction.}
This can occur only if agent $A$ encounters by round $t$ during its exploration a group of at least $f+1$ agents at a node $v\ne u$ that claim being in state {\tt token} of phase $i$. Among these agents there is necessary at least one good agent $B$. Denote by $t'$ the last round before round $t$ such that agent $A$ is not in state {\tt explorer}. According to Algorithm Byz-Unknown-Size, $t'\geq t-T(EST(n_i))-1$, and thus in round $t'$ agent $B$ is either in substate {\tt red} of phase $i$ (at node $v$), or in state {\tt token} of phase $i$ (at node $v$), or in state {\tt explorer} of phase $i$. Now denote by $t''$ the last round such that $t''\leq t'$ and such that agent $A$ or $B$ is in substate {\tt red} of phase $i$. Let us first assume that agent $B$ is in substate {\tt red} of phase $i$ in round $t''$: in this round agent $A$ is then either {\tt red} (at node $u$) or in state {\tt token} (at node $u$), and tests configuration $\phi_i$ in view of Algorithm Byz-Unknown-Size. So still according to Algorithm Byz-Unknown-Size, there are at least $k$ {\tt red} agents that claim testing configuration $\phi_i$ at node $v$ in round $t''$, while there are at least $f+1$ agents (which are {\tt red} or in state {\tt token}) that claim testing configuration $\phi_i$ at node $u$ in the same round $t''$. Hence there are at least $k+1$ good agents which test the same configuration $\phi_i$ in round $t''$ and such that each of them is either {\tt red} or in state {\tt token}. However, this implies that each of these good agents has its label in configuration $\phi_i$ (otherwise it would have been impossible for at least one of them to be in substate {\tt red} of phase $i$ or in state {\tt token} of phase $i$ in round $t''$ according to the rules of state {\tt setup}): there are then at least $k+1$ distinct labels in $\phi_i$ which contradicts the definition of $k$. Hence agent $B$ cannot be in substate {\tt red} of phase $i$ in round $t''$. Let us now consider that agent $A$ is in substate {\tt red} of phase $i$ in round $t''$: in this round agent $B$ is then either {\tt red} (at node $v$) or in state {\tt token} (at node $v$) or in state {\tt explorer} in round $t''$, and tests configuration $\phi_i$. If agent $B$ tests configuration $\phi_i$ and is either {\tt red} or in state {\tt token} in round $t''$, then similarly as above we can get a contradiction with the definition of $k$. If agent $B$ tests configuration $\phi_i$  and is in state {\tt explorer} in round $t''$, then it is either {\tt red} or in state {\tt token} at node $v$ in round $t''-T(EST(n_i))-1$, while agent $A$ is necessary {\tt red} at node $u$ in the same round in view of the definition of $t''$ and Algorithm Byz-Unknown-Size. Therefore similarly as above we can prove there are at least $k+1$ good agents which test the same configuration $\phi_i$ in round $t''-T(EST(n_i))-1$ and such that each of them is either {\tt red} or in state {\tt token} leading again to a contradiction with the definition of $k$. Hence Case~1 is impossible.    

\item{Case~2: in round $t$ agent $A$ adds a node to its BFS tree $T$ under construction that has to be rejected.} 
This can occur only if at some point during the exploration of agent $A$, there are not at least $f+1$ agents at node $u$ that claim being in state {\tt token} of phase $i$. Since agent $A$ adds a node $x$ to its BFS tree, we know that the execution of procedure $EST'$ by agent $A$ does not terminate in round $t$ as it remains to make at least the process of node $x$. From round $t$ on, if agent $A$ does not meet any group of at least $f+1$ agents that claim being in state {\tt token} of phase $i$ during the current execution of $EST'$, then at some point it has spent more than $T(EST(n_i))$ rounds executing $EST'$ in phase $i$ and transits to state {\tt failure} of phase $i$: we get a contradiction with that fact that agent $A$ enters state {\tt tower} of phase $i$ from state {\tt explorer} of phase $i$. So there is a round after round $t$, during its execution of $EST'$ in phase $i$, where agent $A$ meets a group of at least $f+1$ agents that claim being in state {\tt token} of phase $i$. This group is located either on a node different from node $u$ or on node $u$ (in which case a "new token" appears on node $u$ after the disappearance of the first one): by using similar arguments to those used above, we can also get a contradiction with the definition of $k$ in both these situations.
\end{itemize}

So, $n_i$ cannot be different from $n$, and thus the claim follows.

Now we can conclude the proof for the case where agent $A$ enters state {\tt tower} from state {\tt explorer} in round $r$. Indeed, according to Claim~1, the BFS tree $T$ computed by agent $A$ before transiting to state {\tt token} in round $r$ is of size $n$. However, according to Algorithm Byz-Unknown-Size agent $A$ can transit from state {\tt explorer} of phase $i$ to state {\tt token} of phase $i$ only if the size of $T$ is equal to $n_i$. Hence we necessarily have $n=n_i$ in this case.

To end the proof of this lemma, it remains to consider the case where agent $A$ enters state {\tt tower} from state {\tt token} in round $r$. According to Algorithm Byz-Unknown-Size, agent $A$ can make such a transition only if it previously transited from state {\tt explorer} of phase $i$ to state {\tt token} of phase $i$. Hence, still according to Algorithm Byz-Unknown-Size, agent $A$ computed a BFS tree $D$ having size $n_i$, while in state {\tt explorer} of phase $i$. However in view of the claim below, the size of $D$ is necessarily equal to $n$. Hence $n=n_i$, which proves the lemma also holds in the case where agent $A$ enters state {\tt tower} from state {\tt token} in round $r$.

{\bf Claim~2.} The size of $D$ is $n$.

If $n_i=n$ then the claim is true because agent $A$ can transit from state {\tt explorer} to state {\tt token} in round $r$ only if $m=n_i$. Concerning the case where $n_i\ne n$, similarly to what is done for this case in the proof of Claim~1, we can get the same contradictions, which proves the claim.

\end{proof}


\begin{lemma}
\label{lem:lem46}
In any round there is at most one tower.
\end{lemma}

\begin{proof}
Assume by contradiction there exists some round $r$ when there are two distinct towers $\mathcal{D}_j$ and $\mathcal{T}_k$. The members of $\mathcal{D}_j$ (resp. $\mathcal{T}_k$) are in state {\tt tower} and all have index $j$ (resp. index $k$). Note that in the case where $\mathcal{D}_j$ and $\mathcal{T}_k$ are at the same node, indexes $j$ and $k$ are different from each other, otherwise we would have $\mathcal{D}_j=\mathcal{T}_k$ according to the definition of a tower. In the other case, index $j$ is not necessarily different to index $k$.

By Lemma~\ref{lem:lem42}, there is a good agent $A\in\mathcal{D}_j$ (resp. $B\in\mathcal{T}_k$) which has been in state {\tt tower} since round $r-j$ (resp. round $r-k$). Since agent $A$ (resp. agent $B$) is in state {\tt tower} with index $j$ (resp. index $k$) in round $r$, agent $A$ (resp. agent $B$) enters state {\tt tower} in round $r-j$ (resp. round $r-k$) from state {\tt explorer} or {\tt token}. Let $\phi_i$ be the configuration tested by agent $A$ in round $r$. In round $r-j$, agent $A$ enters state tower of phase $i$ from state {\tt explorer} of phase $i$ or state {\tt token} of phase $i$. So, by Lemma~\ref{lem:lem45} we have $n_i=n$, and then we know that $i\geq\alpha$. Hence according to Lemma~\ref{lem:lem43} agent $B$ also tests configuration $\phi_i$ in round $r$, and according to Algorithm Byz-Unknown-Size it enters state {\tt tower} of phase $i$ from state {\tt explorer} or {\tt token} of phase $i$ in round $r-k$. Let $t$ be the last round before round $r$ such that agent $A$ or $B$  is {\tt red} in phase $i$. Without loss of generality, assume that in round $t$ agent $A$ is {\tt red} and tests configuration $\phi_i$. According to Algorithm Byz-Unknown-Size agent $A$ enters state {\tt explorer} or {\tt token} of phase $i$ from state {\tt tower builder} in round $t+1$ at some node $u$. By Lemmas~\ref{lem:lem43} and~\ref{lem:lem44}, in round $t+1$ agent $B$ is either in state {\tt wait-for-a-tower} of phase $i$ or also enters state {\tt explorer} or {\tt token} of phase $i$ from state {\tt tower builder} at node $u$. Let us first consider the first case in which agent $B$ is in state {\tt wait-for-a-tower} of phase $i$ in round $t+1$. In this case, we cannot have $r-k<t+1$ as it is impossible to transit directly or indirectly from state {\tt tower} of phase $i$ to state {\tt wait-for-a-tower} of phase $i$. So, we necessarily have $r-k\geq t+1$. However this implies that agent $B$ cannot enter state {\tt tower} of phase $i$ from state {\tt explorer} or {\tt token} of phase $i$ in round $r-k$ because it is impossible to reach (directly or indirectly) state {\tt explorer} or {\tt token} of phase $i$ from state {\tt wait-for-a-tower} of phase $i$. The first case is therefore impossible.
Concerning the second case, it implies that agents $A$ and $B$ computed the same set $H$ in round $t$ at the same node $u$ (i.e., during the last round of their waiting period as {\tt red} agents of phase $i$). Hence according to Algorithm Byz-Unknown-Size agents $A$ and $B$ enter together state {\tt tower} of phase $i$ in round $t+2|H|T(EST(n))+1=r-j=r-k$, and thus belong to the same tower in round $r$, which is a contradiction.
\end{proof}

\begin{lemma}
\label{lem:lem47}
If a good agent declares gathering is over at node $v$ in round $r$, then all good agents are at node $v$ in round $r$ and declare that gathering is over in round $r$.
\end{lemma}

\begin{proof}
Assume by contradiction there is a good agent $A$ that declares gathering is over at node $v$ in round $r$ but there is a good agent $B$ that does not make such a declaration at the same node and in the same round. Without loss of generality, we assume round $r$ is the first round when an agent declares gathering is over. Thus, agent $B$ does not declare gathering is over before round $r$. Let $\phi_i$ be the configuration tested by agent $A$ in round $r$. Let $k$ be the number of labels in configuration $\phi_i$.

According to Algorithm Byz-Known-Size, there is a tower $\mathcal{T}_{P(n_i)}$ in round $r$ at node $v$: all agents belonging to $\mathcal{T}_{P(n_i)}$ are in state {\tt tower} and have index $P(n_i)$ (which corresponds to the number of edge traversal of procedure $EXPLO(n_i)$). By Lemma~\ref{lem:lem42}, there is a good agent $C$ inside of tower $\mathcal{T}_{P(n_i)}$ which has been in state {\tt tower} since round $r-P(n_i)$. Thus agent $C$ is in state {\tt tower} and has index $0$ in round $r-P(n_i)$. This implies that agent $C$ enters state {\tt tower} from state {\tt explorer} or {\tt token} in round $r-P(n_i)$. 

Let $\phi_j$ be the configuration tested by agent $C$ in round $r-P(n_i)$. By Lemma~\ref{lem:lem45}, we have $n_j=n$, and thus $j\geq\alpha$. Moreover according to Algorithm Byz-Unknown-Size, agent $C$ still tests configuration $\phi_j$ in round $r$. Since agent $A$ tests configuration $\phi_i$ in round $r$ and $j\geq\alpha$, by Lemma~\ref{lem:lem43} we know $i=j$. Hence all good agents in $\mathcal{T}_{P(n_i)}$ in round $r$ test configuration $\phi_i$, which implies all good agents of $\mathcal{T}_{P(n_i)}$ declare that gathering is over in round $r$ according to Algorithm Byz-Unknown-Size. Agent $B$ also tests configuration $\phi_i$ in round $r$ but it is not in state {\tt tower} of phase $i$, as otherwise it would belong to $\mathcal{T}_{P(n_i)}$ by Lemma~\ref{lem:lem46} and thus it would also declare that gathering is over in round $r$, which contradicts the definition of agent $B$.

As mentioned above we have $j\geq\alpha$, $i=j$ and $n=n_j$. Hence $i\geq\alpha$ and $n_i=n$.

According to Lemmas~\ref{lem:lem43} and~\ref{lem:lem44}, when agent $C$ enters state {\tt explorer} of phase $i$ or state {\tt token} of phase $i$ from state {\tt tower builder} of phase $i$ in round $r-P(n)-2|H|*T(EST(n))$ at some node $u$ (for the definition of $H$, refer to the description of substate {\tt red} in Algorithm Byz-Unknown-Size), then in the same round agent $B$ either also enters state {\tt explorer} of phase $i$ or state {\tt token} of phase $i$ from state {\tt tower builder} of phase $i$ at node $u$, or is in state {\tt wait-for-a-tower} of phase $i$.

Let us first consider the first situation. In this situation, agent $B$ determined the same set $H$ as agent $C$ in the last round as a {\tt red} agent in phase $i$. According to Algorithm Byz-Unknown-Size, agent $B$ also enters state {\tt tower} of phase $i$ at the same node and in the same round as agent $C$, or agent $B$ enters state {\tt failure} of phase $i$ by round $r-P(n)$. In the former case, agent $B$ and $C$ then belong to tower $\mathcal{T}_{P(n_i)}$ in round $r$ and agent $B$ declares that gathering is over in round $r$, which is a contradiction. The latter case in which agent $B$ is in state {\tt failure} of phase $i$ by round $r-P(n)$ requires a deeper analysis. According to Algorithm Byz-Unknown-Size, for each round of $\{r-P(n)-2|H|*T(EST(n)),\cdots,r-P(n)-1\}$ agent $C$ is either in state {\tt token} of phase $i$ or in state {\tt explorer} of phase $i$. If agent $B$ and $C$ are in state {\tt token} of phase $i$ when agent $B$ decides to transit to state {\tt failure} of phase $i$, then agent $C$ does the same thing according to Algorithm Byz-Unknown-Size and we get a contradiction with the fact that agent $C$ declares gathering is over in round $r$ as an agent in state {\tt token} of phase $i$ (as it is impossible to reach state {\tt tower} of phase $i$ once in state {\tt failure} of phase $i$). If agent $B$ is explorer while agent $C$ is {\tt token} when agent $B$ decides to transit to state {\tt failure} then that implies agent $B$ computes a BFS tree of a size smaller than $n_i=n$. (Indeed agent $B$ can compute a BFS tree of a size larger than $n_i$ only if at some point its token disappears, but that implies that agent $C$ enters state {\tt failure} of phase $i$ and as above we get a contradiction with the definition of $C$. For the same reasons, agent $B$ cannot spent more than $T(EST(n_i))$ rounds to compute the BFS tree or decide to transit to state {\tt failure} after having backtracked to its token i.e., after having traversed all edges traversed during the execution of procedure $EST'$ in the reverse order and the reverse direction). The only reason is that at some point agent $B$ rejects a node $x$ that has to be added to its BFS tree under construction: similarly as in the proof of Claim~1 in the proof of Lemma~\ref{lem:lem45} we can get a contradiction with the definition of $k$. If agents $A$ and $B$ are in state {\tt explorer} of phase $i$ in the same round we get a contradiction with the fact that they computed the same set $H$ before leaving state {\tt tower builder} of phase $i$. It remains to analyse the case in which agent $B$ is {\tt token} while agent $C$ is {\tt explorer} when agent $B$ decides to transit to state {\tt failure} of phase $i$ in some round $t<r-P(n_i)$. According to Algorithm Byz-Unknown-Size the token of agent $C$ disappears in round $t$. Moreover similarly as in the proof of Claim~1 in the proof of Lemma~\ref{lem:lem45}, we can argue that from round $t$ on, agent $C$ cannot meet a group of at
least $f+1$ agents that claim being in state {\tt token} of phase $i$. Hence when agent $C$ verifies in round $r-P(n_i)-1$ if it is with a group of at
least $f+1$ agents that claim being in state {\tt token} of phase $i$ (just before entering state {\tt tower} in round $r-P(n_i)$), the condition cannot be fulfilled and thus agent $C$ cannot enter state {\tt tower} in round $r-P(n_i)$, which is a contradiction. Since we get a contradiction in all cases, the first situation cannot occur.

Let us now consider the second situation in which agent $B$ is in state {\tt wait-for-a-tower} of phase $i$ in round $r-P(n)-2|H|*T(EST(n))$, while $C$ enters state {\tt explorer} of phase $i$ or state {\tt token} of phase $i$ from state {\tt tower builder} of phase $i$ in the same round. By using similar arguments to those used in the latter two paragraphs of the proof of Lemma~\ref{lem:6}, we can distinguish the case in which agent $B$ enters state {\tt failure} of phase $i$ in some round of $\{r-P(n)-2|H|*T(EST(n))+1,\cdots,r\}$ from the case in which it does not, and argue that we obtain a contradiction in both cases. Hence agent $B$ does not exist and the lemma holds.

\end{proof}

\begin{lemma}
\label{lem:48}
There is at least one good agent that ends up declaring that gathering is over.
\end{lemma}

\begin{proof}
Assume by contradiction no agent ends up declaring gathering is over. Let $\phi_i\in\Omega$ a good configuration which corresponds to the initial configuration of all good agents in the graph. We then have $n_i=n$. Let $k$ be the number of labels in configuration $\phi_i$: by definition $k\geq f+2$. Since configuration $\phi_i$ is good and $i\geq\alpha$, in view of Proposition~\ref{prop:prop42} and Lemma~\ref{lem:lem41} and the fact that at least one good agent is woken up by the adversary, every good agent reaches state {\tt tower builder} of phase $i$ possibly in different rounds but at the same node $v$ corresponding to the node where the agent having the smallest label is initially located. Similarly as in proof of Lemma~\ref{lem:7} we can prove the following claim (by adjusting the waiting period accordingly and by using Propositions~\ref{prop:prop41} and~~\ref{prop:prop42} as well as Lemmas~\ref{lem:lem41} and~\ref{lem:lem43} instead of respectively Propositions~\ref{prop:prop2} and~\ref{prop:prop1} as well as Lemmas~\ref{lem:lem1} and~\ref{lem:lem3}). Round $r$ is the first round when a good agent becomes {\tt yellow} in phase $i$.

{\bf Claim~1.} All good agents are {\tt yellow} at node $v$ and test configuration $\phi_i$ in round $r+T(EXPLO(n))+n+Q_i-1$.

From Claim~1 and Algorithm Byz-Unknown-Size, the $k$ good agents become {\tt red} together before round $r+2T(EXPLO(n))+2n+2Q_i$ and decide in the same round to leave state {\tt tower builder} (to enter state {\tt tower} or {\tt explorer}) after having waited exactly $T(EXPLO(n))+n + T(EST(n)) + Q_i$ rounds as {\tt red} agents of phase $i$. During the last round of this waiting period, all good agents compute the same set $H$ which includes each of their labels (and possibly at most $f$ forged labels of Byzantine agents). Note that $H$ cannot be larger than $n=n_i$ as otherwise that would imply that there are more agents than nodes in the graph. Hence $H\leq n$ and each of the $k$ good agents enters state {\tt explorer} or {\tt token} in the same round, call it $s$, at node $v$. Since set $H$ is the same for all good agents and no two of them have the same labels, in round $s$ there is exactly at most one good agent $A_1$ that enters state {\tt explorer} of phase $i$, while the other good agents enter state {\tt token} of phase $i$. In fact if the smallest label in $H$ corresponds to a label of a good agent there is exactly one good agent $A_1$ that enters state {\tt explorer} of phase $i$, otherwise $A_1$ does not exist and all the good agents enters state {\tt token} of phase $i$. According to the rules of state {\tt token}, in each round of $\{s,s+1,\cdots,s+2T(EST(n))-1\}$ there are at least $f+1$ agents in state {\tt token} of phase $i$ at node $v$ as there are at least $k-1\geq f+1$ good agents in state {\tt token} of phase $i$. As a result, during the execution of procedure $EST'$ by agent $A_1$ as an {\tt explorer} (if it exists), its ``token never disappears'' and it cannot see another token (i.e., a group of at least $f+1$ agents in state {\tt token} of phase $i$) at a node $u\ne v$ because there are at most $f$ Byzantine agents and all the good agents different from explorer $A_1$ are at node $v$. Hence agent $A_1$ computes a BFS tree of size $n$ within $T(EST(n))$ rounds, and then enters state {\tt token} in round $s+2T(EST(n))$ after having backtracked to node $v$. 
Regardless of whether there was a good agent $A_1$ or not, in view of $H$ and the rules of states {\tt explorer} and {\tt token}, there is at most one good agent $A_2$ (which had not entered state {\tt explorer} of phase $i$ yet) that enters state {\tt explorer} of phase $i$ from state {\tt token} in round $s+2T(EST(n))$. More precisely, if the second smallest label in $H$ corresponds to a label of a good agent there is exactly one good agent $A_2$ that enters state {\tt explorer} of phase $i$ in round $s+2T(EST(n))$, otherwise $A_2$ does not exist and all good agents are in state {\tt token} of phase $i$ in round $s+2T(EST(n))$. Moreover, in each round of $\{s+2T(EST(n)),s+2T(EST(n))+1,\cdots,s+4T(EST(n))-1\}$ there are at least $f+1$ agents in state {\tt token} of phase $i$ at node $v$ as there are at least $k-1\geq f+1$ good agents in state {\tt token} of phase $i$ (recall that if there was a good agent $A_1$, it is now in state {\tt token}). Thus, similarly as before, we know that if agent $A_2$ exists, it computes a BFS tree of size $n$ within $T(EST(n))$ rounds, and is at node $v$ with the other good agents in state {\tt token} in round $s+4T(EST(n))-1$ at the latest. From this point, regardless of whether there was a good agent $A_2$ or not, if $H=2$ (by definition $H\geq 2$) then all the good agents enters state {\tt tower} of phase $i$ (including $A_2$ if any) in round $s+4T(EST(n))$ according to Algorithm Byz-Unknown-Size. Otherwise, in view of $H$ and the rules of states {\tt explorer} and {\tt token}, there is at most one good agent $A_3$ (which had not entered state {\tt explorer} of phase $i$ yet) that enters state {\tt explorer} of phase $i$ from state {\tt token} in round $s+4T(EST(n))$, while agent $A_2$ (if any) enters the state in which the other good agents are i.e., state {\tt token} of phase $i$. Again via similar arguments, we know that if $H=3$ then all the good agents enters state {\tt tower} (including $A_3$ if any) in round $s+6T(EST(n))$ according to Algorithm Byz-Unknown-Size. Otherwise, by induction on the number of labels in $H$ we can prove that all good agents enter state {\tt tower} at node $v$ in round $s+2|H|T(EST(n))$. Hence, we know the $k$ good agents will make together an entire execution of procedure $EXPLO(n)$: according to the rules of state {\tt tower}, agent $A$ declares gathering is over at the end of this execution. So, we get a contradiction, and the lemma holds.
\end{proof}

From Lemmas~\ref{lem:lem47} and~\ref{lem:48}, we know that Algorithm Byz-Unknown-Size and Theorem~\ref{theo:theo6} are valid.

\section{Conclusion}

We provided a deterministic $f$-Byzantine gathering algorithm for arbitrary connected graphs of known size (resp. unknown size) provided that the number of good agents is at least $f+1$ (resp. $f+2$). By providing these algorithms, we closed the open question of what minimum number of good agents $\mathcal{M}$ is required to solve the problem, as each of our algorithms perfectly matches the corresponding lower bound on $\mathcal{M}$ stated in \cite{DieudonnePP14}, which is of $f+1$ when the size of the network is known and of $f+2$ when it is unknown. Our work also highlighted the fact that the ability for the Byzantine agents to change their labels has no impact in terms of feasibility when the size of the network is initially unknown, since it was proven in \cite{DieudonnePP14} that $\mathcal{M}$ is also equal to $f+2$ when the Byzantine agents do not have this ability.

While we gave algorithms that are optimal in terms of required number of good agents, we did not try to optimize their time complexity. Actually, the time complexity of both our solutions depends on the enumerations of the initial configurations, which clearly makes them exponential in $n$ and the labels of the good agents in the worst case. Hence, the question of whether there is a way to obtain algorithms that are polynomial in $n$ and in the labels of the good agents (with the same bounds on $\mathcal{M}$) remains an open problem.

\bibliographystyle{plain}
\bibliography{byzantin}

\end{document}